\documentclass[12pt]{article}
\usepackage{setspace}

\usepackage{JASA_manu}
\usepackage{hyperref}
\usepackage[utf8]{inputenc}
\usepackage{amsmath,amsfonts,amsthm}
\usepackage{graphicx}
\usepackage{xcolor}
\usepackage{natbib}
\usepackage[algo2e]{algorithm2e}

\usepackage[normalem]{ulem}

\usepackage{amssymb}
\usepackage{tabu}
\usepackage{caption}
\usepackage{mdwlist}
\usepackage{subfigure}

\usepackage[symbol]{footmisc}

\usepackage{xr}
\newtheorem{prop}{Proposition}
\newtheorem{theorem}{Theorem}
\newtheorem{corollary}{Corollary}
\newtheorem{definition}{Definition}

\newtheorem{lemma}{Lemma}
\begin{document}

\def\spacingset#1{\renewcommand{\baselinestretch}%
{#1}\small\normalsize} \spacingset{1}

\title{Heavy-tailed distribution for combining dependent $p$-values with asymptotic robustness}
\author{
Yusi Fang  \\
Department of Biostatistics \\ University of Pittsburgh, Pittsburgh, PA 15261  \\
email: \texttt{yuf31@pitt.edu}
\and
Chung Chang\footnote{correspond to: cchang@math.nsysu.edu.tw}\\
Department of Applied Mathematics\\
National Sun Yat-sen University,
Kaohsiung, Taiwan\\
email: \texttt{cchang@math.nsysu.edu.tw}
\and 
Yongseok Park \\
Department of Biostatistics \\
University of Pittsburgh, Pittsburgh, PA 15261\\
email: \texttt{yongpark@pitt.edu}
\and
George C. Tseng \\
Department of Biostatistics \\
University of Pittsburgh, Pittsburgh, PA 15261\\
email: \texttt{ctseng@pitt.edu}

}

\maketitle

\newpage
\doublespacing
\mbox{}

\begin{center}
\textbf{Abstract}
\end{center}
    The issue of combining individual p-values to aggregate multiple small effects is a long-standing statistical topic. Many classical methods are designed for combining independent and frequent signals using the sum of transformed p-values with the transformation of light-tailed distributions, in which Fisher’s method and Stouffer’s method are the most well-known. In recent years, advances in big data promoted methods to aggregate correlated, sparse and weak signals; among them, Cauchy and harmonic mean combination tests were proposed to robustly combine p-values under ”arbitrary” dependency structure. Both of the proposed tests are the transformation of heavy-tailed distributions for improved power with the sparse signal. Motivated by this observation, we investigate the transformation of regularly varying distributions, which is a rich family of heavy-tailed distribution, to explore the conditions for a method to possess robustness to dependency. We show that only an equivalent class of Cauchy and harmonic mean tests has sufficient robustness to dependency in a practical sense. We also show an issue caused by large negative penalty in the Cauchy method and propose a simple, yet practical modification with fast computation. Finally, we present simulations and apply to a neuroticism GWAS application to verify the discovered theoretical insights.

\vspace*{.3in}


\newpage

\section{Introduction}
\label{s:intro}
Combining $p$-values to aggregate information from multiple sources is a long-standing issue in social science and biomedical research. Classical methods mostly focus on combining multiple independent and frequent signals to increase statistical power, which can be viewed as a type of meta-analysis. Consider the combination of $n$ independent $p$-values, $\vec{p}=(p_1,...,p_n)$. Many earlier methods were developed in the form of statistics $T(\vec{p})=\sum_{i=1}^n g(p_i)=\sum_{i=1}^n F^{-1}_U (1-p_i)$ to sum up transformed $p$-values, where the transformation $g(p)$ is the inverse CDF of $U$. Conventional methods in this category include Fisher's method \citep{fisher1932statistical} with $T=\sum_{i=1}^n -2 \log(p_i)$ using $U$ as a chi-squared distribution and Stouffer's method \citep{stouffer1949american} with $T=\sum_{i=1}^n -\Phi^{-1}(p_i)$ using $U$ as a standard normal distribution, among many other choices of $g(p)$ and their corresponding $U$ in the literature \citep{edgington1972additive,pearson1933method,mudholkar1979logit}. This first category of methods aims for classical meta-analysis to combine independent and relatively frequent signals and it applies light-tailed distribution (i.e. tails thinner than an exponential function) for $U$. Efficiency of a method is mostly considered under the asymptotic framework that the number of $p$-values $n$ is fixed and sample size $m$ to derive each $p$-value goes to infinity, where $p=O(e^{-m})$ in most cases. Under this setting, it has been shown that only the equivalent class of Fisher's method is asymptotically Bahadur Optimal (ABO), meaning the efficiency of the combined $p$-value statistics is asymptotically optimal under fixed $n$ and $m\rightarrow\infty$ \citep{littell1971asymptotic}.

In the rise of big data, many scientific questions have turned to combine $p$-values with large $n$. The seminal paper by \cite{donoho2004higher} established a framework of combining $p$-values with weak and sparse signals and proposed the higher-criticism test with asymptotically optimal property. This second category of methods considers $n\rightarrow \infty$ and only a small number $s$ of the $n$ p-values ($s=n^\beta$ where $0<\beta<\frac{1}{2}$) have weak signals ($p=O(n^{-r}/\log^{\frac{1}{2}}n)$ with $0<r<1$) while all remaining $p$-values have no signal (i.e. $p\overset{D}{\sim} unif(0,1)$). Under this setting, the classical minimum $p$-value method ($minP$) $T=\mathop{min}_{1\leq i\leq n} p_i$ is asymptotically optimal only for $0<\beta<1/4$ while higher criticism is asymptotically optimal for all possible $0<\beta<1/2$.  Several methods, including Berk-Jones test \citep{berk1979goodness,li2015higher}, were subsequently proposed to improve finite-sample power of higher-criticism while maintaining the asymptotic efficiency.

All aforementioned methods were developed to combine independent $p$-values. Many modern large-scale data analyses have generated the need of combining a large number of dependent $p$-values with sparse and weak signals, which we categorize as methods for the third category. A notable application is to combine $p$-values of multiple correlated SNPs (can be tens to hundreds or thousands) in a SNP-set (e.g. all SNPs in a gene region or in gene regions of a pathway) in genome-wide association studies (GWAS). In this case, the neighboring SNPs often pose varying degrees and unknown dependency structures. Efforts have been made to extend existing tests to account for dependency using permutation or other numerical simulation approaches \citep[e.g.][]{liu2019accurate}. Permutation or simulation-based methods are, however, not practical when $n$ is large and high precision of $p$-value is needed to account for multiple comparison. \cite{barnett2017generalized} developed an analytic approximation for higher criticism incorporated with dependency structure. The method is, however, still computationally intensive and not accurate enough for small $p$-values needed for multiple comparison. Motivated by these needs, \cite{liu2020cauchy} and \cite{wilson2019harmonic} independently proposed Cauchy combination test ($T=\sum_{i=1}^n \tan\{(0.5-p_i)\pi\}$) and harmonic mean combination test ($T=\sum_{i=1}^n \frac{1}{p_i}$) to combine $p$-values under unknown dependency structure. A remarkable property of both methods is that the null distribution and testing procedure derived from independence assumption are robust under dependency structure in an asymptotic but practical sense to be explained later. In this paper, we set out to explore a rich family of transformation $g(p)$ from their corresponding $U$ (i.e., the regularly varying distribution family) and investigate the conditions such that practical robustness to dependency similar to Cauchy and harmonic mean methods can be achieved. We note that selections of $U$ for classical meta-analysis setting (fixed $n$ and $m\rightarrow\infty$) are all from thin-tailed distributions (e.g. chi-squared distribution for Fisher's method and Gaussian for Stouffer's method). This is reasonable since a thin-tailed distribution produces evener contributions from marginally significant $p$-values in the meta-analysis of frequent signals. In contrast, Cauchy and harmonic mean methods correspond to heavy-tailed distributions of $U$, which highly focus on small $p$-values and down-weigh marginally significant $p$-values. Figure \ref{transformations} shows the transformation function of $g(p)$ in log-scale. For Fisher's method, the contributions of $p$-values $10^{-2}$ and $10^{-6}$ to the test statistics are $4.6$ and $13.8$. For heavy-tailed transformation methods, the contributions become $100$ versus $10^6$ for harmonic mean and $31.82052$ versus $3.18\times10^5$ for Cauchy. With an increased focus on small $p$-values, the methods are more powerful in detecting sparse signals. It is worth noting that the recent work by \cite{vovk2020combining} also considered the sum of transformed $p$-values to combine $p$-values and showed an upper bound of significance level inflation under general dependence structure. We will describe the difference between our results and theirs in detail in the remark following Theorem \ref{thm2}.

Throughout this paper, when we call a thin-tailed, heavy-tailed or regularly varying method, it means that its corresponding $U$ is a thin-tailed, heavy-tailed or regularly varying distribution. The paper is structured as the following. We first investigate Box-Cox transformation for $g(p)$ in Section \ref{s:connection}, which is equivalent to Pareto distribution for $U$. In Section \ref{section2.1}, we will build connection and insight of existing methods including $minP$, harmonic mean, Cauchy and Fisher in this framework. Particularly, we show that the Cauchy method is approximately equivalent to the harmonic mean method, which is a special case of the Box-Cox transformation. In Section \ref{section2.2}, we observe that the Cauchy method can potentially suffer from the large negative penalty for $p$-values close to 1. We introduce a simple, yet practical solution using truncated Cauchy with fast computing. In Section \ref{section3}, we will introduce a family of heavy-tailed distribution, namely regularly varying distribution, and investigate the conditions in the family that can provide robustness for dependency structure as in Cauchy and harmonic mean (Section \ref{section3.1}-\ref{section 3.2}). Section \ref{section3.3} shows the asymptotic power and detection boundary under the sparse and weak alternatives considered in \cite{donoho2004higher}. Section \ref{section4} contains extensive simulations to demonstrate type I error control and power of different methods and numerically verify the theoretical results. Section \ref{section5} contains a GWAS application of neuroticism to compare the performance of different methods and demonstrate the improvement of the truncated Cauchy method over the Cauchy method. Section \ref{section6} provides the final conclusion and discussion.

\section{Connection between minP, harmonic mean, Cauchy and Fisher}
\label{s:connection}
\subsection{Methods by Pareto distribution to connect four existing methods}\label{section2.1}
As mentioned in Section \ref{s:intro}, we observe that many methods for the first category to combine independent and relatively frequent $p$-values all correspond to thin-tailed distributions for $U$ and many methods for the second and third categories for combining sparse and weak signals utilize heavy-tailed distributions. In this subsection, we consider Pareto distribution for $U$, which is equivalent to Box-Cox transformation for $g(p)$. We will build the connection of four existing methods: $minP$, harmonic mean, Cauchy and Fisher, based on this transformation family. Insight in Pareto distribution also provides intuition when we introduce the regularly varying distribution as an extended richer family in the next section. Finally, we will prove the approximate equivalency of the harmonic mean and Cauchy combination methods.
Consider the family of $p$-value combination methods: $T=\sum_{i=1}^n g(p_i)$, where $g(p)= \frac{1}{p^{\eta}}$ for some $\eta>0$. We can show that $g(p)=F^{-1}_U(1-p)$ such that $U\overset{D}{\sim} Pareto(\frac{1}{\eta}, 1)$. In other words, $P(U>t)=t^{-\frac{1}{\eta}}$ for $t>1$, which means $U$ is a heavy-tailed distribution. A larger $\eta$ corresponds to a heavier tail. Particularly, the harmonic mean method corresponds to $\eta=1$ in Pareto distribution. We note that, by denoting $\lambda=-\eta$, we can rewrite $h(p; \lambda)=\frac{g(p;\eta)-1}{\lambda}=\frac{p^\lambda-1}{\lambda}$, which is Box-Cox transformation. The following Proposition \ref{prop1} shows that minP and Fisher are limiting cases in the Pareto distribution when $\eta\rightarrow +\infty$ and when $\eta\rightarrow 0$. Proposition \ref{prop2} shows that the Cauchy combination method is approximately identical to harmonic mean for relatively small $p$-values.

\begin{prop}\label{prop1}
For fixed $n$, $minP$ is a limiting case of methods by Pareto distribution when $\eta\rightarrow \infty$. Similarly, the Fisher's method is the limiting case of Pareto when $\eta\rightarrow 0$.
\end{prop}
\begin{proof}
Denote by $T_{\gamma_m}=\sum_{i=1}^n \frac{1}{p_i^{\gamma_m}}=\sum_{i=1}^n \frac{1}{p_{(i)}^{\gamma_m}}$, where $p_{(i)}$’s are ordered $p$-values. Note that $T_{\gamma_{m}}$ is equivalent to $T_{\gamma_m}^*=\left(\sum_{i=1}^n \frac{1}{p_i^{\gamma_m}}\right)^{\frac{1}{\gamma_m}}=\frac{1}{p_{(1)}}\left(\sum_{i=1}^n \left(\frac{p_{(1)}}{p_{(i)}}\right)^{\gamma_m}\right)^{\frac{1}{\gamma_m}}$. As $\gamma_m\rightarrow\infty$, $T_{\gamma_m}^*\rightarrow \frac{1}{p_{(1)}}$, which is equivalent to $minP$.

To prove the result of Fisher's method, note that $T_{\gamma_m}$ is equivalent to $T_{\gamma_m}^{**}=\sum_{i=1}^n\frac{p_i^{-\gamma_m}-1}{-\gamma_m}$. By L'Hospital's rule , we have $\lim_{\gamma_m\rightarrow 0}\frac{p ^{-\gamma_m}-1}{-\gamma_m}=\log (p)$. Hence $T_{\gamma_m}^{**}\rightarrow \sum_{i=1}^n \log(p_i)$ almost surely and is equivalent to the Fisher's method.
\end{proof}

\begin{prop}\label{prop2}
The Cauchy combination test is approximately identical to harmonic mean for relatively small $p$-values in the sense that, $ \frac{\pi \cdot g^{(CA)}(p)-g^{(HM)}(p)}{ g^{(HM)}(p)}=O(p^2)$.
\end{prop}
\begin{proof}
By Taylor's expansion, $g^{(CA)}(p)=\tan\left\{(0.5-p)\pi\right\}\approx\frac{1}{\pi p}-\frac{\pi p}{3}-\frac{(\pi p)^3}{45}+\cdots$. The result immediately follows.
\end{proof}

It is somewhat surprising that even though the forms of transformation of Cauchy and harmonic mean are quite different, they are approximately equivalent and the behavior of both can be characterized by the index $\eta=1$ of the Box-Cox transformation. It is natural to ask if there exist other $p$-value combination methods in an extended rich heavy-tailed distribution family to enjoy similar finite-sample robustness property as in the Cauchy and harmonic mean methods. To answer this question, we introduce the family of regularly varying distribution and investigate the properties in Section \ref{section3}.

Figure \ref{transformations} shows minus log-scaled $p$ transformation $g(p)$ versus minus log-scaled transformation $g(p)$ for the $BC_{0.5}$ (i.e. Box-Cox transformation with $\eta=0.5$), $HM$ (the harmonic mean method, equivalent to $BC_1$), $CA$ (the Cauchy method), $BC_{1.5}$, Fisher's and Stouffer's methods. We see that as $\eta$ increases, smaller $p$-values will be more dominant and impact of marginally significant $p$-values rapidly diminishes, which gives stronger power for sparse signal applications. $CA$ and $HM$ are approximately proportional when $p$ sufficiently small (roughly when $p<10^{-2}$).

\begin{figure}
\includegraphics[scale=0.45]{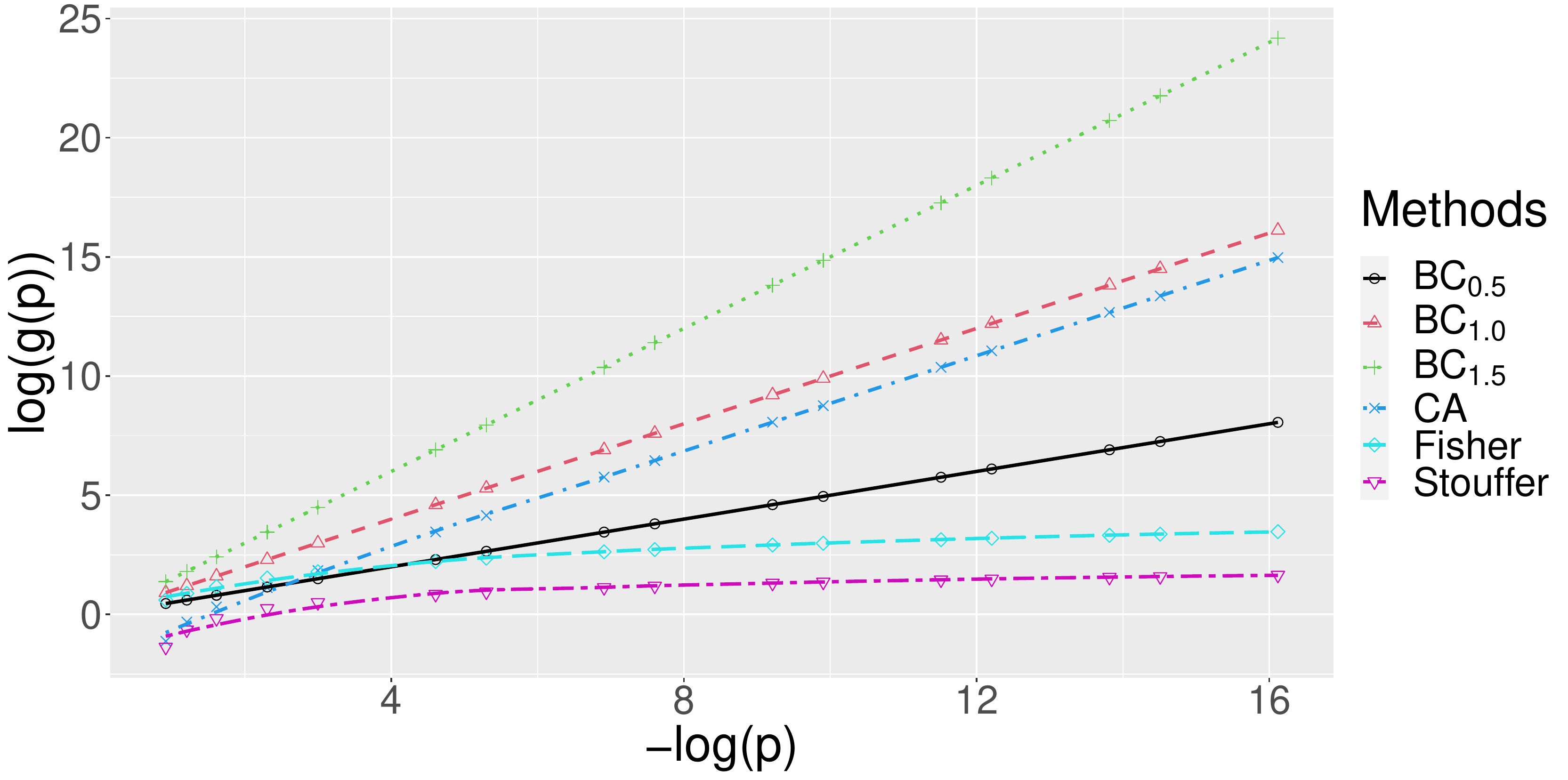}
\caption{
Comparison of transformations. We show 6 different transformations of $p$-values, $g(p)$, which correspond to $BC_{0.5}$, $BC_{1}$ ($HM$), $BC_{1.5}$, $CA$, Fisher  and Stouffer. The x-axis is $-\log(p)$ and the y-axis shows $\log(g(p))$.
}
\label{transformations}
\end{figure}

\subsection{Large negative penalty issue in Cauchy and a truncated Cauchy modification}\label{section2.2}
We have shown that $HM$ and $CA$ are approximately equivalent and simulations in Section \ref{section4.2} will confirm their almost identical performance. We note that when a $p$-value is very close to $1$, the contribution in the Cauchy method is close to negative infinity, which can potentially cause numerical issues and substantial power loss. The situation of a $p$-value closes to $1$ can happen frequently for tests of discrete data, in which case the $p$-values under null hypothesis may not necessarily be $unif(0,1)$.Two other possible situations to cause $p$-values close to $1$ are when $n$ is large or when the model to derive $p$-values are mis-specified.
As a simple remedy, we propose a truncated Cauchy test ($CA^{tr}$) that truncates any of the $n$ $p$-values greater than $1-\delta$  to be $1-\delta$. For example, when $\delta=0.01$, we have $p^{tr}= p$ if $p<0.99$ and $p^{tr}= 0.99$ if $p\geq 0.99$. The proposed method can also be viewed in the form of summation of transformed $p$-values. Indeed, the statistic of $CA^{tr}$ can be written as:
\begin{align}
T_{CA^{tr}}=\sum_{i=1}^n \tan\left(\pi\left(\frac{1}{2}-p_i\right)\right)1(p_i<1-\delta)+\tan\left(\pi\left(\delta-\frac{1}{2}\right)\right)1(p_i\geq 1-\delta).   \notag
\end{align}

The theorems to be introduced in Section \ref{section3} imply that $CA^{tr}$ enjoys almost the same advantages of the Cauchy method in terms of type I error control and power for the detection of weak and sparse signals.
The test statistic of $CA^{tr}$ no longer follows the standard Cauchy distribution under the null assumption. To deal with the computational issue of the truncated Cauchy method, we propose a hybrid strategy, which uses approximation by generalized central limit theorem (GCLT) in general but switches to an efficient importance sampling procedure by cross-entropy parameter selection when $n$ is small ($n<25$) and the targeted size is large ($\alpha\geq 5\times 10^{-3}$).  

Below we first show that when $n$ is sufficiently large, we can apply generalized central limit theorem (GCLT) from \cite{shintani2018super} to approximate the null distribution of $T_{CA^{tr}}$ below.  

\begin{prop}\label{GCLT}
Let $\nu_\delta=\tan\left(\pi(\delta-\frac{1}{2})\right)$, $f_{1n}=\int_{\nu_\delta}^{+\infty}\frac{\cos(x/n)}{(1+x^2)}$, $f_{2n}=\int_{\nu_\delta}^{+\infty}\frac{\sin(x/n)}{(1+x^2)}$ and $\theta_n=\arctan\left(\frac{\delta\sin(\nu_\delta/n)+((1-\delta)/\pi)f_{2n}}{\delta\cos(\nu_\delta/n)+((1-\delta)/\pi)f_{1n}}\right)$. Then we have:
\begin{align}
    \frac{T_{CA^{tr}}-n^2\theta_n}{n}\stackrel{d}{\longrightarrow} S(1,1,\frac{1}{2},0),\notag
\end{align}
where $S(\alpha,\beta,\gamma,\mu)$ is a stable distribution with parameters $\alpha=1,\beta=1,\gamma=\frac{1}{2}$ and $\mu=0$,
which is defined with its characteristic function as: 
\begin{align}
  S(x ; \alpha, \beta, \gamma, \mu)=\frac{1}{2 \pi} \int_{-\infty}^{\infty} \phi(t) e^{-i x t} \mathrm{~d} t , \notag
\end{align}
with $\phi(t) =\exp \left\{i \mu t-\gamma^{\alpha}|t|^{\alpha}(1-i \beta \operatorname{sgn}(t) w(\alpha, t))\right\} $ and 
\begin{align*}
w(\alpha, t)=\left\{\begin{array}{ll}
\tan (\pi \alpha / 2) & \text { if } \alpha \neq 1 \\
-2 / \pi \log |t| & \text { if } \alpha=1.
\end{array}\right..
\end{align*}
\end{prop}

\noindent \underline{Remark}: 

Proposition \ref{GCLT} can be obtained by simple calculation using formula ($4$) in \cite{shintani2018super}. Table S1 examines the approximation performance of GCLT for small $n$ and varying size $\alpha$. The result shows satisfying accuracy when $\alpha<5\times 10^{-3}$. When $\alpha \geq 5\times 10^{-3}$, GCLT needs larger $n$ to perform well (roughly $n\geq25$). As a result, we develop an efficient importance sampling procedure for this scenario. Briefly, Proposition \ref{upperboundtrCAU} below gives narrow upper and lower bounds for the tail probability of truncated Cauchy. By applying the framework proposed by \cite{de2005tutorial} for estimating rare event probability, we develop a cross-entropy procedure to search within the narrow bounds for a high-precision approximation for the tail probability of the truncated Cauchy. Details of the efficient importance sampling are shown in Supplement Section S2.2. Table S1 further shows the accurate calculation of the importance sampling with affordable computing when $n<25$. In summary, when calculating p-values for $CA^{tr}$, to balance the computing and performance, we propose to set $\delta=0.01$ and use GCLT approximation when $\alpha< 5\times 10^{-3}$ or $n\geq 25$. When $\alpha\geq 5\times 10^{-3}$ and $n< 25$, importance sampling will be used. In Section \ref{section4.3} and Section \ref{section5}, we will demonstrate the superior performance of truncated Cauchy over Cauchy using simulations and a real application. Specifically, it avoids the large negative penalty issue of the Cauchy method but still enjoys similar robust properties for type I error control under dependency and power for detecting weak and sparse signals.  

\begin{prop}\label{upperboundtrCAU}
Let $1-\delta$ be the truncation point of truncated Cauchy test. The upper tail probability of the null distribution of the truncated Cauchy method satisfies:
\begin{align}
P\left(X_1\geq t\right)\leq P\left(T_{CA^{tr}}>t\right)\leq P\left(X_1\geq t\right)\left(1+\delta\right)^n,   \notag
\end{align}
where $X_1$ is a Cauchy distributed random variable.
\end{prop}

\section{Asymptotic properties of regularly varying methods for $p$-value combination}\label{section3}
\subsection{Regularly varying tailed distribution}\label{section3.1}
Before introducing the regularly varying distributions, we first define some notations.
Throughout this paper, denote by $\bar{F}$ the survival function of the distribution $F$ (i.e., $\bar{F}(t)=1-F(t)$ for any $t$). Limits and asymptotic properties are assumed to be for $t\rightarrow \infty$ unless mentioned otherwise.  For
two positive functions $u(\cdot)$ and $v(\cdot)$, we write $u(t)\sim v(t)$ if $\lim_{t\rightarrow \infty}\frac{u(t)}{v(t)} = 1$.
Also, if $\lim_{t\rightarrow \infty}\frac{u(t)}{v(t)}>1$, we write $u(t)\gtrsim v(t)$; if $\lim_{t\rightarrow \infty}\frac{u(t)}{v(t)}<1$, we write $u(t) \lesssim v(t)$. The definition of regularly varying tailed distribution is given below:
\begin{definition} A distribution $F$ is said to belong to the regularly varying tailed family with index $\gamma$ (denoted by $F\in R_{-\gamma}$) if
$$\lim_{x\rightarrow\infty}\frac{\bar{F}(xy)}{\bar{F}(x)}=y^{-\gamma}$$
for some $\gamma > 0$ and all $y>0$. 
\end{definition}
We denote the whole family of regularly varying tailed distributions as $R$.
It can be shown that every distribution $F$ belonging to $R_{-\gamma}$ can be characterized
by $$\bar{F}(t)\sim L(t)t^{-\gamma},$$
where $L(t)$ is a slowly varying function. A function $L$ is called slowly varying if $\lim_{y\rightarrow\infty}\frac{L(ty)}{L(y)}=1$ for any $t>0$.
Some examples of slowly varying functions $L(t)$ are $1, \ln(t)^{\nu}, \ln(\ln(t))$. Given the property of slowly varying function $L(t)$, the tail of regularly varying distribution converges to zero at a relatively slow rate, which leads to the heavy-tailed property.  

The regularly varying tailed family includes many interesting distributions: Pareto distribution, Cauchy distribution, log-gamma distribution and inverse gamma distribution. Indeed, the survival function of Pareto(a,b)
is $\bar{F}(t)=\frac{b}{t^a},\;t>b$ and hence
$U\in R_{-a}$. In addition,
the survival function of Cauchy distribution is
$\bar{F}(t)\sim\frac{1}{t\pi}$ and therefore $U\in R_{-1}$.

An important property for regularly varying tailed distributions is as follows:
Assume $U_1,\ldots,U_n$ are i.i.d. random variables with distribution function $F\in R_{-\gamma}$. Then
\begin{equation}\label{independence}
   P(U_1+\ldots+U_n>t)\sim nP(U_1>t).
\end{equation}

\subsection{Asymptotic tail probability approximation and robustness to dependence}\label{section 3.2}
The first theorem below investigates the approximation of the null distribution of the test statistic.
Assume that the $p$-values are obtained from z-scores; that is, all the test statistics follow normal distributions.
Specifically, let $\mathbf{X}=(X_1,\ldots,X_n)$ be the random vector (z-scores) for the $n$ test statistics. The mean of $\mathbf{X}$ is $\mathbf{\mu}=(\mu_1,\ldots,\mu_n)$ and correlation matrix $\mathbf{\Sigma}$.  Since we can always rescale test statistics, we assume each $X_i$ has variance 1. Under the null hypothesis, $H_0: \mu_i=0,\forall i=1,\ldots,n$, hence the $p$-value for the $i$th study is $p_i=2(1-\Phi(|X_i|))$ for $i=1,\dots,n$. Recall from the introduction section, we consider the test statistic $T(\mathbf{X})= \sum_{i=1}^n g(p_i)= \sum_{i=1}^n g(2(1-\Phi(|X_i|)))$, which is a sum
of transformed $p$-values. When $p_i\overset{D}{\sim} unif(0,1)$ under the null hypothesis, $g(p_i)$ is a random variable, where we denote $g(p_i)\overset{D}{\sim} U$, which is consistent with previously introduced relationship $g(p_i)=F^{-1}_U(1-p_i)$ when $U$ is a continuous random variable. We further assume the following conditions for $T(\mathbf{X})$:\\
\textbf{(A1)} $\forall 1\le i<j\le n$, $X_i$ and $X_j$ are bivariate normally distributed.\\
\textbf{(A2)} Let $U_i=g(p_i),i=1,\ldots,n.$ with $U_i\overset{D}{\sim} U \in R_{-\gamma}$ under $H_0$. Assume the function $g(p)$ is continuous and $g(p)$ satisfies one of the two situations: (A2.1) $g(p)$ is strictly decreasing in $(0,1)$; (A2.2)  $g(p)$ is bounded below (i.e., $g(p)>c’$ for certain constant $c’$) and is strictly decreasing on in $(0,c)$ with some constant $0<c<1$. \\
\textbf{(A3)} (\textit{balance condition}) Under $H_0$, let $F$ be the CDF of $U$ and $G(t)=P(|U|>t)=t^{-\gamma}L(t)$ where $L(t)$ is a slow-varying function. Assume $\frac{\bar{F}(t)}{G(t)}\rightarrow p$ and $\frac{F(-t)}{G(t)}\rightarrow q$ as $t\rightarrow\infty$,
where $0< p\le 1$ and $p+q=1$. \\

Condition (A1) is mild and is also assumed in \cite{liu2020cauchy} when investigating the robustness of the Cauchy method under arbitrary correlation structure. In fact, this condition is to guarantee the tail distributions of each pair of $U_i$ and $U_j$ are asymptotically independent; see the precise definition of asymptotically tailed independence for a pair of random variables in the Supplement.

Condition (A2) includes the Box-Cox transformation (satisfying A2.1), Cauchy transformation (satisfying A2.1) and truncated Cauchy transformation (satisfying A2.2)  introduced in Section \ref{section2.2}. Condition (A3) is called
"balance condition", which is a common condition for regularly varying tailed random variables \citep{goldie1998subexponential}. 
For example, for
the harmonic mean method, $p=1,q=0$; for the Cauchy method, $p=q=1/2$, and for the truncated Cauchy method, $p=1,q=0$.
\begin{theorem}\label{thm 1}
Under conditions (A1), (A2) and (A3) and assume
$\rho_{ij},1\le i<j\le n$, the $(i,j)$th element of $\mathbf{\Sigma}$, satisfies $-1<\rho_{ij}<1$.
Then under $H_0:\mathbf{\mu}=\mathbf{0}$ and for any correlation matrix $\mathbf{\Sigma}$,
We have
$$P(T(\mathbf{X})>t)\sim n P(U>t).$$
\end{theorem}
Here $T(\mathbf{X})=\sum_{i=1}^n U_i$ is the sum of correlated regularly varying tailed random variables. 
The theorem is somewhat surprising and a general result since it is applicable to any regularly varying method and any correlation structure $\mathbf{\Sigma}$ with $-1<\rho_{ij}<1$ as long as no perfect correlation exists.
This theorem is essentially based on Theorem 3.1 in \cite{chen2009sums}, i.e. Lemma S2 in the Supplement.  
Roughly speaking, because of the heaviness of the tail for each $U_i$ and the asymptotic tailed independence between each pair of $U_i$ and $U_j$, asymptotically the correlation structure has very limited influence on the tail of T(X). 
Since the approximated tail probability is independent of $\mathbf{\Sigma}$, an immediate application is to derive the $p$-value of a regularly varying method under independence assumption (i.e. $P(U_1+\cdots +U_n>t)$ with i.i.d. $U_1,\cdots ,U_n$; see Equation \eqref{independence}). The theorem warrants its asymptotic robustness to arbitrary dependence structure as similarly shown in the harmonic mean and Cauchy methods \citep{wilson2019harmonic,liu2020cauchy}. Or alternatively, one may approximate the tail probability by $nP(U>t)$. We, however, note that the robustness to arbitrary dependence structure is in an asymptotic sense, meaning extremely large $t$ (corresponding to extremely small test size $\alpha$) may be required for different tail heaviness in $U$ and correlation structure to guarantee a good approximation.

\begin{figure}
\hspace*{-2cm}\includegraphics[scale=0.40]{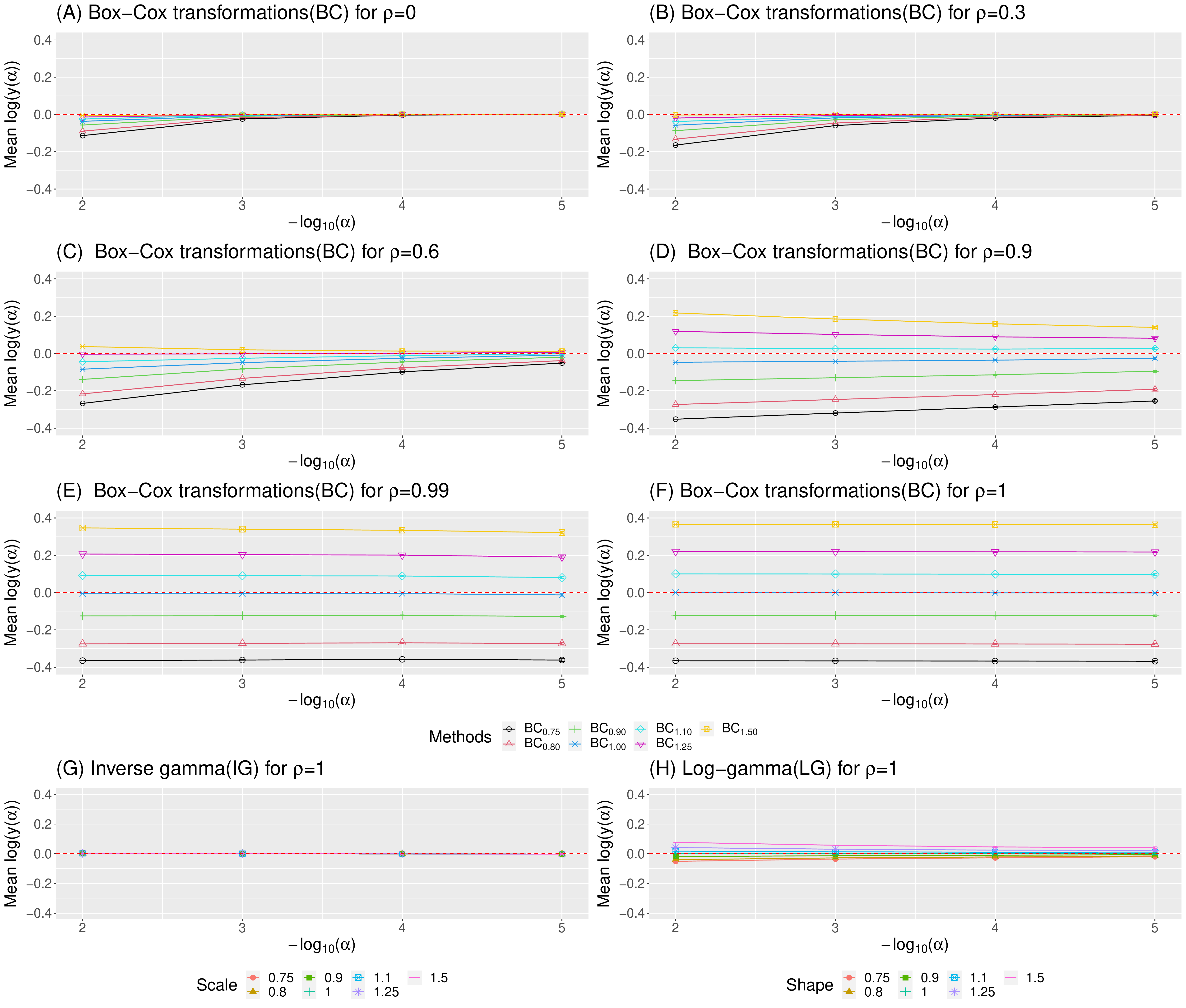}
\caption{The mean log-scaled $y(\alpha)$ for Box-Cox transformations, inverse gamma and log-gamma across different significance levels $\alpha$.
(A)-(F) represent the results of Box-Cox transformations with different values of $\eta=$ 0.75, 0.8, 0.9, 1, 1.1, 1.25, 1.5 for correlation level $\rho=$0, 0.3, 0.6, 0.9, 0.99 and 1 respectively.
(G) represents the results of inverse gamma with shape parameter equals 1 and the values of scale parameter across 0.75, 0.8, 0.9, 1, 1.1, 1.25, 1.5 for correlation level $\rho=1$.
(H) represents the results of log-gamma with rate parameter equals 1 and the values of scale parameter across 0.75, 0.8, 0.9, 1, 1.1, 1.25, 1.5 for correlation level $\rho=1$.
The x-axis is the negative logarithm of significance level $\alpha$ to base 10 where
$\alpha$ is set to be $10^{-2},10^{-3},10^{-4},10^{-5}$ and the red dash line is the reference line $\log(y(\alpha))=0$ for all the sub-figures.
}
\label{ratio_corr}
\end{figure}

Below we perform a simple simulation to demonstrate and investigate Theorem \ref{thm 1}. Assume $n=3$ and $\mathbf{X}=(X_1,X_2,X_3)$ is multivariate normal with unit variance and common pairwise correlation $\rho_{ij}=\rho$ ($1\le i<j\le 3$). In this simulation we set $\rho=$ 0, 0.3, 0.6, 0.9 and 0.99.
Here we consider 7 Box-Cox tests, $BC_{0.75}$, $BC_{0.8}$, $BC_{0.9}$, $BC_{1}$, $BC_{1.1}$, $BC_{1.25}$, and $BC_{1.5}$. From Theorem \ref{thm 1}, we calculate $y(\alpha)=\frac{nP(U>t_{\alpha})}{P(T(\mathbf{X})>t_{\alpha})}$ from simulations, where $t_{\alpha}$ is chosen so that $P(T(\mathbf{X})>t_{\alpha})=\alpha$ and $\alpha=10^{-2},10^{-3},10^{-4},10^{-5}$.
We expect $\lim_{t_{\alpha}\rightarrow \infty} \log\left(y(\alpha)\right)=0$ when $-1<\rho<1$. Figure \ref{ratio_corr}A-\ref{ratio_corr}E  show $\log_{10}$-scale $\alpha$ on the x-axis and the mean $\log\left(y(\alpha)\right)$ on the y-axis for different $\rho=$(0, 0.3, 0.6, 0.9, 0.99). We note that, as $\rho$ increases, smaller $\alpha$ will be required to observe a good approximation. Theorem \ref{thm2} below further characterizes what would happen if partial of the p-values have perfect correlations $\rho_{ij}=1$ or $-1$.

\begin{theorem}\label{thm2}
Suppose the conditions (A1), (A2) and (A3) in Theorem \ref{thm 1} hold. Define an arbitrary weight vector $w=(w_1,\cdots,w_n)\in R^n_+$, $T_{n,w}=\sum_{i=1}^n w_i g(p_i)$. Also assume
$\rho_{ij}=1 \text{ or }-1$ for $ 1\le i< j\le m$, and $|\rho_{ij}|<1$ for $i>m$ or $j>m$. We have:
\[P(T_{n,w}\left(\mathbf{X})>t\right)\sim\left\{\left(\sum_{i=1}^m w_i\right)^{\gamma}+\sum_{i=m+1}^n w_i^{\gamma}\right\}P(U>t).\]
\end{theorem}

Consider a special case $w=(1,\cdots,1)$. An immediate consequence of Theorem \ref{thm2} is that only when $\gamma=1$ (e.g., $HM$ or $CA$ or $CA^{tr}$ method) can satisfy $\{(\sum_{i=1}^m w_i)^{\gamma}+\sum_{i=m+1}^n w_i^{\gamma}\}=m^\gamma+(n-m)=n$, which produces the asymptotic robustness of Theorem \ref{thm 1}. In other words, Figure \ref{ratio_corr}A-\ref{ratio_corr}E already shows a hint that the convergence of Theorem \ref{thm 1} becomes more and more difficult when $\rho$ increases to almost 1. When some of the p-values have perfect correlation, only index $\gamma=1$ of the regularly varying distribution can still enjoy the asymptotic robustness to arbitrary dependence structure. Figure \ref{ratio_corr}F shows an simulation with $\rho=1$, which satisfies the condition of Theorem \ref{thm2}. By assuming $w_1=w_2=w_3=1$ and $\rho=1$, we have $P(T_{n,w}\left(\mathbf{X})>t\right)\sim  3^{\gamma}P(U>t)$. Figure \ref{ratio_corr}F verifies Theorem \ref{thm2} that only $BC_1$ can reach the convergence $\lim_{t_{\alpha}\rightarrow \infty} \log\left(y(\alpha)\right)=0$, showing robustness to perfect correlation. Although Figure \ref{ratio_corr}E ($\rho=0.99$) and Figure \ref{ratio_corr}F ($\rho=1$) visually look almost identical, all $BC$ methods in Figure \ref{ratio_corr}E will eventually converge to 0 as $\alpha\rightarrow 0$ by Theorem \ref{thm 1}, although very slowly. On the other hand, in Figure \ref{ratio_corr}F, only $BC_1$ can converge to 0 by Theorem \ref{thm2}.

\begin{corollary}\label{cor1}
Suppose the conditions in Theorem \ref{thm2} hold and assume $\sum_{i=1}^n w_i=n$, then we have:
\begin{align}
\left\{\begin{array}{lr}
P(T_{n,w}(\mathbf{X})>t)&\sim n P(U>t) \quad \text{if $\gamma=1$,}\\
P(T_{n,w}(\mathbf{X})>t)&\gtrsim  nP(U>t)\quad \text{if $\gamma>1$,}\\
P(T_{n,w}(\mathbf{X})>t)&\lesssim nP(U>t)\quad \text{if $\gamma<1$.}
\end{array}\right.\notag
\end{align}
\end{corollary}

From Corollary \ref{cor1}, note that, when $w_1=\cdots=w_n=1$ and the transformation $g(p)=1/p^{1/\gamma}$, the test statistic $T_{n,w}$ corresponds to
the statistic $BC_{\eta},\eta=1/\gamma$. Hence, the BC tests with $\eta<1$ (i.e., $\gamma>1$) are anti-conservative in this situation; the higher the value of $\gamma$ is, the more anti-conservative the test is.  This is verified by Figure \ref{ratio_corr}F for $BC_{0.9}$, $BC_{0.8}$ and $BC_{0.75}$ when $\rho=1$.
As $\eta\rightarrow 0$ (i.e., $\gamma\rightarrow\infty$), $BC_{\eta}$ is asymptotically equivalent to the Fisher's method and is the most anti-conservative under dependence.
On the other hand, for $\eta>1$ (i.e., $\gamma<1$), all the corresponding tests $BC_{\eta}$ ($\eta>1$) are conservative under this dependence structure, which is confirmed by Figure \ref{ratio_corr}F for $BC_{1.1}$, $BC_{1.25}$ and $BC_{1.5}$.
In particular, when $\eta\rightarrow\infty$ ($\gamma\rightarrow 0$), $BC_{\eta}$ becomes  $minP$, which hence is expected to be very conservative. Figure \ref{ratio_corr}G and \ref{ratio_corr}H verifies that since inverse gamma and log-gamma are also regularly varying distributions with index $\gamma=1$, they enjoy asymptotic robustness to correlation structure similar to $HM$ ($BC_1$) and Cauchy even when perfect correlation exists.

\begin{corollary}\label{cor2}
If we further assume $-1<\rho_{i,j}<1,\forall 1\le i<j\le n$ (i.e., $m=0$), then we have $$P(T_{n,w}>t)\sim \sum_{i=1}^n w_i^{\gamma}P(U>t).$$
\end{corollary}

Corollary \ref{cor2} shows that, among regularly varying methods, only methods with index $\gamma=1$ are robust to weights. Also note that this formula can be considered to be an extension of Corollary 1.3.8 in \citep{mikosch1999regular}, in which $U_1,\ldots,U_n$ are assumed to be independent regularly varying distributed random variables.

\noindent \underline{Remark}: 

Note that the robustness property of Theorem \ref{thm 1} and \ref{thm2} is similar to \citep{liu2020cauchy,wilson2019harmonic} and only describes the asymptotic behavior of the tail probability of our proposed family. Indeed, the results of Theorem \ref{thm 1} and \ref{thm2} only guarantee that the type I errors of the corresponding tests ($\gamma=1$, equivalent to harmonic mean and Cauchy) can be well controlled for a small size $\alpha$ given fixed $n$ and $\mathbf{\Sigma}$. Intuitively, as $n$ increases, a more stringent cutoff corresponding to a small $\alpha$ is needed to ensure the robustness of type I error control. An ideal robustness property should be to achieve a uniform upper tail bound in the sense of $P(T(X)>t_{\alpha})\leq c\cdot \alpha$ under any dependence structure $\mathbf{\Sigma}$, where $t_\alpha$ is the tail threshold when a nominal $\alpha$ is controlled under independence assumption, and $c$ is independent of $n$ and $\mathbf{\Sigma}$ and is in a reasonable magnitude (e.g., $c=1.5$, meaning the inflation of the type I error is at most 50\% in the worst scenario). This uniform bound is, however, not achievable in general. \cite{vovk2020combining} recently provided a remarkable uniform bound for arbitrary dependency structure but dependent on $n$ for the $HM$ method: 
$$
P(HM>t)\le n a^{HM}_n P\left(U>t\right)=\frac{na^{HM}_n}{t}, \text{where $U \overset{D}{\sim} $ $Pareto(1,1)$, }
$$
where the adjusted factor $\alpha_n^{HM}$ is between $\log(n)$ and $e\cdot \log(n)$ (see Proposition 6 in the paper). This bound is, however, not practical in general applications since, considering $n=100$ or 1000, the inflation bound $\alpha_n^{HM}\geq \log(n)$ is at least 4.6 or 6.9 folds. Furthermore, the factor $\alpha_n^{HM}$ is in comparison to type I error in perfect correlation situation (i.e., $\rho=1$), instead of nominal size $\alpha$ under independence. On this issue, \cite{goeman2019harmonic} pointed out an extreme case that when $n=10^5$ and $\mathbf{\Sigma}$ has exchangeable correlation $\rho=0.2$, $HM$ has more than three folds of type I  error inflation (true type I error=0.164 under nominal $\alpha=0.05$).
In Section \ref{section4.1}, we will perform extensive simulations for a wide range of $n$ and size $\alpha$ to investigate the limitation and develop a practical guidance for applying the $HM$ method in daily applications.

\subsection{Asymptotic power}\label{section3.3}
In this subsection, we investigate the asymptotic power and detection boundary of the test $T(\mathbf{X})$
under sparse alternatives as $n\rightarrow\infty$. 
Consider testing the null hypothesis $H_0: \mathbf{\mu}=(\mu_1,\cdots,\mu_n)=\Vec{0}$ for the bivariate normal $\mathbf{X}$. For the alternative, we consider the conventional "weak" and "sparse" signals setting in \cite{donoho2004higher} by assuming a small number of the $n$ signals are non-zero with $|\mu_i|=\sqrt{2\tau\log(n)}$ for $i\in S=\left\{1 \leq i \leq n: \mu_{i} \neq 0\right\}$ with $|S|=s$ and $0<\tau<1$, and the rest $\mu_i=0$ for $i\in S^c$. In addition, the sparsity of the signals is at the order of $s=n^{\beta}$ with $0<\beta<\frac{1}{2}$.

For Theorem \ref{thm3} below, in addition to the conditions (A2) and (A3), we need two additional conditions:\\
\noindent \textbf{Condition (C1)}: We assume $\mathbf{X}\overset{D}{\sim} N(\mathbf{\mu},\mathbf{\Sigma})$
and assume $\mathbf{\Sigma}$ is a banded correlation matrix; i.e.,
its $(i,j)th$ element $\rho_{ij}=0$ for any $|i-j|>d_0$ for some positive constant $d_0>0$.\\
\noindent \textbf{Condition (C2)}: There exist $h\ge 0$ and $t_1>0$ such that
$$\frac{1}{t^{\gamma}(\ln(t))^{h}}\le\bar{F}(t)\le \frac{(\ln(t))^{h}}{t^{\gamma}}$$ for all $t>t_1$.

Condition (C2) is for tail probability of $U_i$ and is a mild condition because $\bar{F}(t)=P(U_i>t)=\frac{L(t)}{t^{\gamma}}$ ($L(t)$ is a slowly varying function).
This condition holds for all the commonly used distributions we have mentioned so far with regularly varying tails with index $\gamma$. In the Supplement, we show that the $BC$, Cauchy and truncated Cauchy methods all satisfy Condition (C2).

\begin{theorem}\label{thm3}
Under conditions (A2), (A3), (C1) and (C2), for any $0<\gamma\le 1$, any significance level $0<\alpha<1$, and $\tau$ satisfying $\sqrt{\tau}+\sqrt{\beta}>1$, then
under the alternative hypothesis we have:
\[\lim_{n\rightarrow \infty}P(T(\mathbf{X})>t_{\alpha})=1,\]
where $t_\alpha$ is the $p$-value cutoff.
\end{theorem}
Theorem \ref{thm3} states that the power of this test $T(\mathbf{X})$ converges to 1 for any
significance level $\alpha>0$ and $0<\gamma\le 1$, or equivalently, that the sum of Type I and II errors goes to zero given the set-up. Indeed, Theorem \ref{thm3} implies that the
methods with $0<\gamma\le 1$ attain the optimal detection boundary defined in \cite{donoho2004higher} in the strong sparsity situation $0<\beta<1/4$.
\cite{liu2020cauchy} showed a similar result for their proposed Cauchy's test. As described in Section \ref{s:connection}, the Cauchy distribution has
regular-varying tail with index $\gamma=1$. This theorem is valid for methods of regularly varying tailed distributions with index
$0<\gamma\le 1$.
Therefore, this theorem can be considered to be a generalization of Theorem \ref{thm3} in \cite{liu2020cauchy}.

\section{Simulations}\label{section4}
In this section, we perform simulations to compare the robustness performance of different $p$-value combination methods under varying correlation levels among $p$-values to verify theoretical results in Section \ref{s:connection} and \ref{section3}. We include 7 methods discussed in Section \ref{s:connection}, $minP$, $BC_{1.25}$, $CA$, $CA^{tr}$, $HM(BC_1)$, $BC_{0.75}$ and the Fisher's method, as well as $HC$ (Higher criticism) and $BJ$ (Berk-Jones test). Section \ref{section4.1} firstly evaluates the type I error control of different methods under independence and varying level of correlation to verify the robustness of $HM$ and Cauchy methods. Further, since the robustness in Theorem \ref{thm2} for $HM$ and Cauchy is an asymptotic result, we further investigate the type I error control for $HM$ under a wide range of $n$, $\rho$ and $\gamma$ to ensure that the robustness of $HM$ and Cauchy is preserved and useful in a practical sense. Section \ref{section4.2} assesses the statistical power under different dependency structures and sparsity of signals in the alternative hypothesis. In Section \ref{section4.3}, we will evaluate the improvement of the truncated Cauchy method over the Cauchy method in a discrete data simulation.

\subsection{Type I error control}\label{section4.1}
In this subsection, we first simulate $n=100$, $X=(X_1,\cdots,X_n)\overset{D}{\sim} N(0, \mathbf{\Sigma})$, $p_i=2(1-\Phi(|X_i|))$ and $T=\sum_{i=1}^n g(p_i)$ for different aforementioned methods. We also assume that $\mathbf{\Sigma}$ has unit variance on the diagonal line and is exchangeable with common correlation $\rho=cor(X_i, X_j)$ for $1\leq i \neq j\leq n$, where $\rho$ is evaluated at 0 (independence), $0.3$, $0.6$ , $0.9$ and $0.99$. Table S2 shows the type I error of the 9 methods with different levels of correlations at $\alpha=0.001$ using $10^6$ simulations under the null hypothesis. As expected all methods control type I error perfectly under independence assumption (i.e., $\rho=0$). When correlation among p-values exists, we find that $minP$ is the most conservative in type I error control followed by $BC_{1.25}$, as expected from the theoretical result in Corollary \ref{cor1}. $CA$, $CA^{tr}$ and $HM$ remain with perfect type I error control in all correlation settings, showing robustness to dependency structure. Fisher and $BJ$ are the most anti-conservative methods in the presence of correlation, followed by slight anti-conservativeness for $HC$ and $BC_{0.75}$.

It is worth noting that according to Theorem \ref{thm 1} and \ref{thm2} for regularly varying distribution transformation, the tail probability $P(T(X)>t)$ under dependence can be asymptotically approximated by that under independence. However, the asymptotic result only guarantees the dependence robustness for very large $t$ (or equivalently very small $\alpha$). We also expect that larger $n$ will require larger $t$ (smaller $\alpha$) to ensure a good approximation. Specifically, \cite{goeman2019harmonic} has pointed out that, with $\rho=0.2$ and $n=10^5$, the much inflated type I error of 0.164 is obtained for size $\alpha=0.05$. Therefore, it is of interest to explore the robustness property of $T(X)$ for dependence in $HM$ for varying $n$, $\alpha$ and $\rho$ to provide a practical guidance in real applications. In Table S3, we extended the simulation for $HM$ with $n=(25, 50, 100, 500, 1000, 2000, 10000)$, $\alpha=(0.05, 0.01, 0.001, 0.0001)$, and $\rho=(0, 0.3, 0.6, 0.9, 0.99)$. Given the combination of $\alpha$ and $n$, we calculated the maximum percent of inflation (PI) across different $\rho$, which is defined as $\text{PI}=(\max_{\rho} \text{type I error} -\alpha)/(\alpha)$. The result confirms the theoretical result that larger $n$ will generate greater type I error inflation under dependence for a fixed $\alpha$ and will require much smaller $\alpha$ to improve the type I error inflation. For example, when $\alpha=0.01$, we have $PI=30\%$ for $n=25$ compared to $PI=80\%$ for $n=10,000$. On the other hand, when $n=10,000$, $PI$ decreases from $80\%$ to $49\%$ when $\alpha$ decreases from 0.01 to 0.0001. In general, the result shows robust type I error control under varying correlation levels in a practical sense when $n\leq 1,000$ and $\alpha\leq 0.05$ with the maximum $PI=50\%$, which inflates type I error from $\alpha=0.01$ to $0.015$ at $n=1000$ and $\rho=0.3$. Even when $n$ increases to 10,000, $PI$ only minimally increases to $80\%$. When multiple comparison is needed such as in the GWAS applications, small $\alpha$ is targeted and the robust type I error control for $HM$ is generally achieved in a practical sense. However, if a single test is performed with a very large $n$, caution should be taken for the type I error inflation (e.g., type I error is 0.072 for $\alpha=0.05$ when $n=10,000$ and $\rho=0.3$).

\subsection{Statistical power}\label{section4.2}
In this subsection, we follow the simulation setting in Section \ref{section4.1} to evaluate statistical power using different methods under different correlation $\rho$ and strengths of the signal. Following the sparse and weak signal setting in \cite{donoho2004higher}, we design the $n$ signals $\mu=(\mu_1,\cdots,\mu_n)$ to contain $n-s$ with no signal ($\mu_{s+1}=\cdots=\mu_n=0$) and the first $s$ have non-zero signals $\mu_1=\cdots=\mu_s=\mu_0=\frac{\sqrt{4\log(n)}}{s^{0.1}}$, where $s/n =(5\%, 10\%, 20\%)$. Section 4.2.1 will compare the power of different methods under varying correlation $\rho$, where the rejection threshold is obtained from the independence assumption and uncorrected for dependence. In Section 4.2.2, we further demonstrate the power comparison of different methods, where the rejection threshold is corrected with precise type I error control under dependency. We note that the correction is only applicable in simulations and are generally not accessible unless extensive permutation test or simulation-based methods are applied.

\subsubsection{4.2.1 \underline{Power comparison with uncorrected rejection threshold from independence assumption}}\quad\label{section4.2.1}

In Section \ref{section4.1}, $BJ$, $HC$, $BC_{0.75}$ and Fisher's method are anti-conservative when rejection threshold from independence assumption is used. In other words, the methods lose control of type I error when the dependence structure exists. As a result, we will only compare $HM$, $CA$, $CA^{tr}$, $BC_{1.25}$ and $minP$ in this subsection to evaluate the power of different methods in varying level of correlation $\rho$. Table S4 shows the power of the five methods. As expected, the statistical power decreases as $\rho$ increases. $HM$, $CA$ and $CA^{tr}$ methods have almost identical power and are superior to $BC_{1.25}$. $minP$ is the least powerful method among the five. Different proportions of signals give similar patterns and conclusions.

\subsubsection{4.2.2 \underline{Power comparison with corrected rejection threshold considering dependence structure}
}\quad\label{section4.2.2}

Since methods except for $CA$, $CA^{tr}$ and $HM$ are either conservative or anti-conservative in type I error control under the presence of correlation, the power comparison in the previous subsection is not completely fair. Here, we evaluate power using the rejection threshold corresponding to the accurate type I error control in each method under each correlation setting. We note that this comparison is theoretically a fairer comparison with accurate type I error control but, on the other hand, is less practical in applications unless the dependency structure is known or computationally intensive approaches are applied to precisely control the type I error.

Table \ref{tab:Corrected} shows results of all $9$ methods. We order the methods by the index $\eta$ of Box-Cox transformation as introduced in Section \ref{s:connection}: $minP$, $BC_{1.25}$, $HM$, $CA$, $CA^{tr}$, $BC_{0.75}$, Fisher, and then add $HC$ and $BJ$ for comparison. We first observe almost identical results of $CA$, $CA^{tr}$ and $HM$, and decreasing power when $\rho$ increases, as expected. We next compare the five methods $minP$, $CA/CA^{tr}/HM$ and Fisher with varying proportion of signals and $\rho$. When $\rho=0$, Fisher is the least powerful when $s/n=5\%$ (power$=0.640$) but becomes more powerful than $CA/CA^{tr}/HM$ and $minP$ when $s/n=10\%$ and $20\%$, showing its superior performance in frequent signals. $CA/CA^{tr}/HM$ consistently have good power in between $minP$ and Fisher. When $\rho$ increases, Fisher quickly drops to almost zero power even with accurate type I error control. For each given $s/n$, $minP$ is slightly less powerful than $CA/CA^{tr}/HM$ at small $\rho$ but becomes much more powerful than $CA/CA^{tr}/HM$ when $\rho$ is large. This is reasonable since at a very high correlation (e.g., $\rho=0.99$), all signals can almost be viewed as coming from one source so taking the smallest $p$-value gives sufficiently complete information. For $BC_{0.75}$ and $BC_{1.25}$, we observe that the performance of $BC_{1.25}$ is generally intermediate in between $minP$ and $CA/CA^{tr}/HM$, and $BC_{0.75}$ is between $CA/CA^{tr}/HM$ and Fisher. We next compare $HC$ and $BJ$ to the other methods. Although these two methods lose control of type I error under dependency structure and are not the focus of this paper, we are curious about their power performance if correlation structure is correctly considered with Type I error control. As shown in Table \ref{tab:Corrected}, $BJ$ is surprisingly powerful for all three proportion of signals when $\rho=0$ (e.g., power$=0.91$ compared to power$=0.640-0.778$ for the other $7$ methods when $s/n=5\%$). But similar to the Fisher's method, $BJ$'s power quickly drops to almost $0$ with the existence of dependency. The power of $HC$ is generally similar to $CA/CA^{tr}/HM$ but becomes weaker than $CA/CA^{tr}/HM$ for larger $\rho$. Both $HC$ and $BJ$ lose much power when $\rho$ increases. One possible explanation is that both tests compare the ordered $p$-values $p_{(i)}$ with the reference value $i/n$, which is not the correct reference under null with dependence structure.

\subsection{Simulation for the large negative penalty issue in the Cauchy method}\label{section4.3}
As discussed in Section \ref{section2.2}, $p$-values close to $1$ lead to large negative penalties in the Cauchy method, which can cause significant power loss. Below, we design a Fisher's exact (hypergeometric) test for a $2\times2$ contingency table to illustrate the issue and evaluate the improvement of the truncated Cauchy method.

We firstly evaluate type I error similar to Section 4.1. We randomly generate $n=20$, $2\times 2$ contingency tables with fixed row and column margins being $200$. The table has only one degree of freedom, assuming it is the upper-left cell of each table undetermined. Under the null hypothesis, rows and columns are independent and we generate the value of the upper-left cell from $Hypergeometric(400, 200, 200)$. We then apply Fisher's exact test to the simulated data of each table and combine the $n=20$ $p$-values using $HM$ and $CA$ methods. We repeat the simulation for $10^5$ times, set significance level at $\alpha=0.05, 0.01, 0.005, 0.001, 0.0005$ and $0.0001$,  and calculate the proportions of rejections at each $\alpha$. As shown in Table \ref{tab:Discrete} (effect size $p_{11}=0$), the type I errors for $HM$ is slightly smaller than the desired significance level under the null hypothesis (e.g. $0.00077$ versus $0.001$) while those for $CA$ are much lower (e.g. $0.00016$ versus $0.001$). The main reason of the conservativeness in both tests is that the null distribution under the simulation setting is skewed towards 1, instead of $unif(0,1)$, in which case $CA$ is more sensitive since it penalizes more for $p$-values close to $1$. As shown in table \ref{tab:Discrete}, the type I error control of $CA^{tr}$ under $\delta=0.01$ is largely improved for all different $\alpha$; e.g., type I error is now $0.00077$, identical to $HM$, when $\alpha=0.001$.   

We next evaluate power for $HM$ and $CA$. Similar to Section \ref{section4.2}, we simulate $10^5$ Monte Carlo samples. All settings are identical to the last paragraph for type I error control except that we now generate $2\times2$ tables with row-column correlation. We first simulate $Y$ from $Hypergeometric(400, 200, 200)$ under independence assumption. We then simulate $Z\overset{D}{\sim} Bin(200-Y, p_{11})$ and take $Y+Z$ as the value for the upper-left cell. We note that $p_{11}=0$ corresponds to the original null hypothesis and the larger effect size $p_{11}$, the stronger signal. We set $p_{11}=0.2$ and $0.3$ and the powers under different $\alpha$ are shown in Table \ref{tab:Discrete}. As expected, larger $p_{11}$ generates higher power for both $HM$ and $CA$. $CA$ produces much smaller power than $HM$ mainly due to impact from skewed $p$-values toward $1$. $CA^{tr}$ largely alleviates the issue and can perform almost identical to $HM$.



\begingroup

\begin{table}
\caption {Mean corrected power for tests Fisher, $BC_{0.75}$, $CA$, $CA^{tr}$(truncated Cauchy), $HM$, $BC_{1.25}$, $minP$, $HC$ and $BJ$ across correlation $\rho=0,0.3,0.6,0.9, 0.99$ and proportion of signals $s/n=5\%,10\%,20\%$. The standard errors are far less than the mean power and hence omitted.} 

\begin{tabular}{llllllll}
\hline\hline
$s/n$ & Methods 
        &  $\rho=0$   & $\rho=0.3$ & $\rho=0.6$ & $\rho=0.9$ &     $\rho=0.99$ \\    \hline\hline
&$Fisher$ & 0.640 & 0.0039 & 0.0021 & 0.0017 & 0.0016 \\  
&$BC_{0.75}$ & 0.778 & 0.615 & 0.437 & 0.308 & 0.269 \\
&$CA$ & 0.749 & 0.620 & 0.490 & 0.387 & 0.348 \\
$5\%$& $CA^{tr}$  & 0.749 & 0.621 & 0.490 & 0.388 & 0.348\\
&$BC_1(HM)$ & 0.749 & 0.621 & 0.491 & 0.389 & 0.348 \\
&$BC_{1.25}$ & 0.735 & 0.618 & 0.509 & 0.438 & 0.402 \\
&$minP$ & 0.712 & 0.603 & 0.522 & 0.532 & 0.600 \\
\hline
&$HC$ & 0.760 & 0.623 & 0.415 & 0.216 & 0.195 \\
&$BJ$ & 0.912 & 0.0015 & 0.0001 & 0.001 & 0.001 \\
\hline\hline
&$Fisher$ & 0.992 & 0.013 & 0.0044 & 0.003 & 0.003 \\
&$BC_{0.75}$ & 0.908 & 0.689 & 0.461 & 0.301 & 0.258 \\
&$CA$ & 0.870 & 0.680 & 0.503 & 0.365 & 0.320 \\
$10\%$&$CA^{tr}$ & 0.870 & 0.681 & 0.503 & 0.366 & 0.319\\
&$BC_1(HM)$ & 0.869 & 0.681 & 0.504 & 0.366 & 0.319 \\
&$BC_{1.25}$ & 0.850 & 0.672 & 0.517 & 0.407 & 0.361 \\
&$minP$ & 0.814 & 0.646 & 0.520 & 0.480 & 0.514  \\
\hline
&$HC$ & 0.887 & 0.691 & 0.432 & 0.213 & 0.206 \\
&$BJ$ & 0.998 & 0.017 & 0.001 & 0.001 & 0.001 \\
\hline\hline
&$Fisher$ & 1.000 & 0.0745 & 0.017 & 0.009 & 0.008\\
&$BC_{0.75}$ & 0.982 & 0.752 & 0.484 & 0.300 & 0.255 \\
&$CA$ & 0.955 & 0.728 & 0.511 & 0.347 & 0.299 \\
$20\%$& $CA^{tr}$ & 0.955 & 0.729 & 0.512 & 0.348 & 0.299 \\
&$BC_1(HM)$ & 0.955 & 0.729 & 0.512 & 0.349 & 0.299 \\
&$BC_{1.25}$ & 0.936 & 0.713 & 0.518 & 0.378 & 0.329 \\
&$minP$ & 0.895 & 0.678 & 0.511 & 0.429 & 0.436 \\
\hline
&$HC$ & 0.973 & 0.749 & 0.451 & 0.227 & 0.231 \\
&$BJ$ & 1.000 & 0.202 & 0.016 & 0.008 & 0.013\\
 \hline\hline
\end{tabular}

\label{tab:Corrected}
\end{table}
\endgroup

\begingroup
\begin{table}
\caption { Mean proportion of rejection of $CA$, $HM$ and $CA^{tr}$(truncated $CA$) across $\rho_{11}=0(\text{type I error}), 0.2(power), 0.3(power)$. The standard errors are far less than the mean proportion and hence omitted.} 

\begin{tabular}{lllllllll}
\hline\hline
$\rho_{11}$ & Methods/Cutoff 
        &  0.05 & 0.01 & 0.005 & 0.001 & $5\times10^{-4}$ & $10^{-4}$                                                        
                                                   \\    \hline
&$CA$ & 0.00825 & 0.00182 & 0.000862 & 0.00016 & 0.0000687 & 1e-05 \\
$\rho_{11}=0$&$BC_1(HM)$ & 0.0386  & 0.00894 & 0.00417  & 0.00077 & 0.000334  & 0.0000487\\
&$CA^{tr}$ & 0.0285 & 0.00729 & 0.00417 & 0.00077 & 0.0000334 & 0.0000487 \\
\hline
&$CA$ & 0.333 & 0.202 & 0.146 & 0.0582 & 0.0408 & 0.0135\\
$\rho_{11}=0.2$&$BC_1(HM)$ & 0.863 & 0.525 & 0.379 & 0.154 & 0.108 & 0.0357\\
&$CA^{tr}$ & 0.848 & 0.522 & 0.377 & 0.154 & 0.108 & 0.0361\\
\hline
&$CA$ & 0.431 & 0.428 & 0.420 & 0.355 & 0.310 & 0.190\\
$\rho_{11}=0.3$&$BC_1(HM)$ & 1.000 & 0.992 & 0.972 & 0.822 & 0.717 & 0.440\\
&$CA^{tr}$ & 1.000 & 0.991 & 0.971 & 0.822 & 0.716 & 0.440\\
 \hline\hline
\end{tabular}

\label{tab:Discrete}
\end{table}
\endgroup
\section{Application}\label{section5}
We apply the $HM$, $CA$, $CA^{tr}$, and $minP$ tests to analyze a GWAS of neuroticism \citep{okbay2016genetic}, a personality trait characterized by easily experiencing negative emotions. The dataset contains $6,524,432$ genetic variants (SNPs) across $179,811$ individuals and $p$-values are calculated for all SNPs to represent the association between the variant and neuroticism. We use genome annotations to locate the genic or intergenic region for each variant. The total number of intergenic and genic regions is $78,895$. Within each genic or intergenic region, we combine $p$-values of variants in each region using the $HM$, $CA$, $CA^{tr}$ and $minP$ methods and obtain the combined $p$-values. Figure \ref{manhattan} shows three Manhattan plots for the combined $p$-values using the $HM$, $CA$ and $minP$ methods, respectively.
As shown in Figure \ref{manhattan}, the combined $p$-values using $CA$ and $HM$ are almost identical and they are slightly more significant than those obtained from $minP$. The bottom right plot in Figure \ref{manhattan} shows the numbers of significant genic or intergenic regions with significance thresholds determined by the Bonferroni procedure (controlling the family-wise error rate at $0.05$) and the FDR procedure (controlling the false discovery rate at $0.05$), or $p$-value threshold at $10^{-4}$, $10^{-5}$ or $10^{-6}$. In all different significance thresholds, the numbers of statistically significant genes for $HM$ and $CA$ are almost identical and they are generally larger than those from $minP$. Particularly, $HM$ and $CA$ both identify $750$ regions under FDR$=5\%$ while $minP$ only finds $476$ regions. 

We input the 750 regions identified by $HM/CA$ under FDR$=5\%$ to the Ingenuity Pathway Analysis package for pathway enrichment analysis. The top enriched pathways include NEUROD1 and NEUROG2, which are transcription factors with important functions in neurogenesis. The top diseases and causal networks identify "neurological disease", which is related to neuroticism. In contrast, by applying the pathway analysis to the top $456$ regions by $minP$, we do not find enriched pathways potentially related to neuroticism and the top causal network is MKNK1, which has not been found to play a role in neurological functions.

We next investigate two regions, SLC2A9 and PCSK6, with small combined $p$-values by $HM$ $p=9.534\times 10^{-4}$ for SLC2A9 and $p=1.527\times 10^{-3}$ for PCSK6; $q$-values $q$=0.0759 for SLC2A9 and $q$=0.0939 for PCSK6) but not by $CA$ ($p=0.9999$ and $0.9999$ and q-values both equal 1). The SLC29A9 gene has been found related to Alzheimer's disease and PCSK6 is related to structural asymmetry of the brain and handedness. We suspect the difference of $HM$ and $CA$ comes from $p$-values close to 1 as described in Section \ref{section4.3}. Figure S2 shows two jitter plots of $p$-values for SNPs in genes SLC2A9 (right) and PCSK6 (left). Both of these two genes contain multiple SNPs with very small $p$-values (e.g. 17 SNPs with $p<10^{-4}$ in SLC2A9 and 8 SNPs for PCSK6) so the gene regions could potentially be significant. But since both genes also contain many SNPs with $p$-values close to 1 (5 SNPs with $p>0.99$ for SLC2A9 and 9 SNPs for PCSK6), $CA$ is impacted and produces larger combined $p$-values than $HM$, a situation similar to that described in Section \ref{section4.3}. Since there are above $500$ $p$-values to combine for both genes, by applying $CA^{tr}$ at $\delta=0.99$ with approximation by GCLT (Proposition \ref{GCLT}), the $p$-values improve to $9.531\times 10^{-4}$ for SLC2A9 and $1.532\times 10^{-3}$ for PCSK6, which are almost identical to the $p$-values calculated by $HM$.

\begin{figure}
\hspace*{-2cm}\includegraphics[scale=0.6]{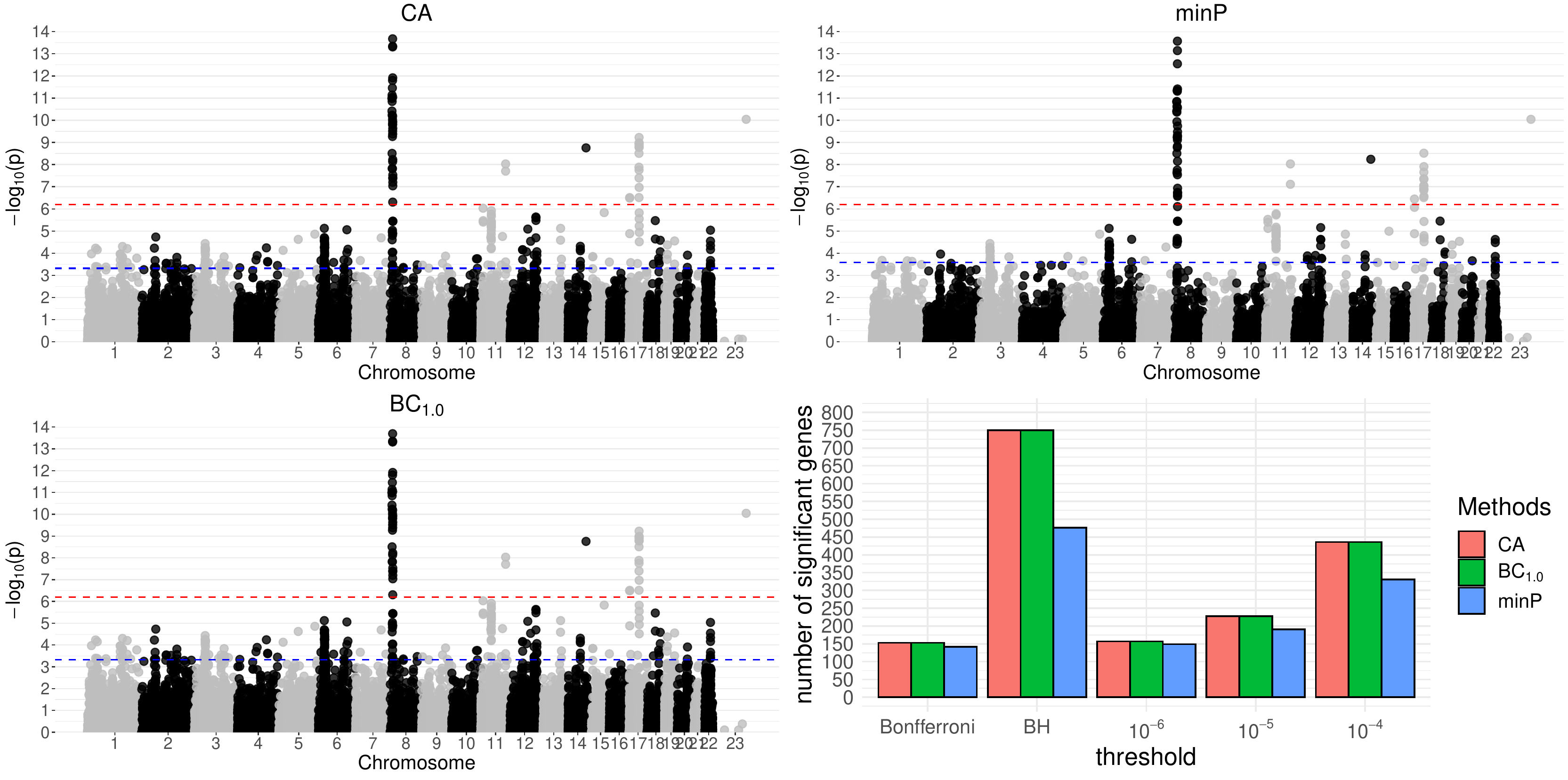}
\caption{Mahattan plots and number of significant $p$-values for $CA$, $BC_1(HM)$ and $minP$. The red dash lines are the cutoffs of Bonferroni correction for $\alpha=5\%$ and the blue dash lines are the cutoffs of Benjamini-Hochberg correction for FDR=$5\%$. The significant regions (FDR=5\%) detected by $HM$ and $CA$ are the same except two regions, DDX58 ($q$=0.0499 by $CA$ and $q$=0.0501 by $HM$) and POU2F3 ($q$=0.0509 by $CA$ and $q$=0.0492 by $HM$).
}
\label{manhattan}
\end{figure}

\section{Discussion}\label{section6}
In this paper, we investigate methods for combining dependent $p$-values using transformation corresponding to regularly varying distribution, which is a rich family of heavy-tailed distribution and includes Pareto distribution (Box-Cox transformation) as a special case. We first present the issue of aggregating multiple $p$-values in three major historical scenarios: (1) classical meta-analysis of combining independent and frequent signals (e.g. Fisher), (2) methods for aggregating independent weak and sparse signals (e.g. $minP$, higher criticism and Berk-Jones), and (3) recent methods for combining $p$-values with sparse signals and unknown dependency structure (i.e. Cauchy and harmonic mean). We then examine popular methods designed for these three settings under the Pareto and regularly varying distribution to provide theoretical insight and finally present the condition of heavy-tailed transformation methods to have the robustness with dependency structure. 

Our contributions are fourfold in both providing theoretical insight and practical application guidelines. Firstly, in Section \ref{s:connection}, we use the family of Box-Cox transformation, or equivalently transformation by CDF of Pareto distributions, to provide connections among Fisher, $CA$, $HM$ and $minP$ methods that are designed to specialize in the three scenarios. We also show that the two recent methods -- $CA$ and $HM$ -- are approximately identical. Secondly, in Section \ref{section3}, we focus on the dependent $p$-value scenario and investigate the condition for $p$-value combination methods under regularly varying distribution to have the robustness to dependency structure, where $CA$ and $HM$ are special cases. We show that only methods of the equivalent class of $CA$ and $HM$ (i.e., index $\gamma=1$) in the regularly varying distribution have the robustness property. Thirdly, we demonstrate an occasional drawback of the Cauchy method when some $p$-values are close to 1, which contributes large negative penalty and causes power loss. We propose a simple, yet practical solution by a truncated Cauchy method with fast and accurate computation. Finally, the simulations and a real GWAS application confirm the theoretical insights and provide a practical guideline for using the harmonic mean and Cauchy methods. Specifically, Table S3 in Section \ref{section4.1} gives guidance of the degree of possible type I error inflation of the harmonic mean method under varying $n$ (number of combined p-values), $\rho$ (correlation level between p-value) and $\alpha$ (test size). 

Modern data science faces challenges from larger data dimension, increased structural complexity, and the need for models and inference to tailor for the subject domain. The three categories of $p$-value combination methods have motivated the development of numerous methods in the literature and is a good example of how statistical theories can provide insight into method development and guide towards real applications. In our paper, we conclude that the condition in regularly varying distribution to have dependency structure robustness in $p$-value combination is those distributions with index $\gamma=1$, which includes Cauchy and harmonic mean methods recently proposed. For future direction, it is of interest whether other methods (e.g. inverse gamma or log-gamma family)  satisfying this condition may enjoy robustness and simultaneously obtain better statistical power in some applications of interest.

\section*{Acknowledgements}
YF and GCT are funded by NIH R21LM012752; CC is funded by Ministry of Science and Technology of ROC 109-2118-M-110-002.
The authors thank Zhao Ren for multiple inspiring discussions.

\section*{Supplementary material}
 Supplementary material includes additional simulation results, as well as details of the efficient importance sampling procedure for the truncated Cauchy method.


\bibliographystyle{apalike} 
\bibliography{references}

\appendix
\section{Technical Arguments: Proof of Theorems}
Before we show the technical arguments, we first define some notations. 

Two nonnegative non-identically distributed random variables $Y_1$ and $Y_2$ with distributions $F_1$ and $F_2$, respectively, are said to be asymptotically tailed independent if
\begin{equation}\label{quasi_ind}
   \lim_{t\rightarrow \infty}\frac{P(Y_1>t,Y_2>t)}{\bar{F}_1(t)+\bar{F}_2(t)}=0.
\end{equation}
It suffices to show the asymptotically tailed independent by showing
$P(Y_1>t|Y_2>t)=o(1)$ or $P(Y_2>t|Y_2>t)=o(1)$, or  equivalently,
$P(Y_1>t,Y_2>t)=o(P(Y_1>t))\text{ or }o(P(Y_2>t))$.

More generally, two real-valued random variables, $Y_1$ and $Y_2$, are said to be asymptotically independent if the relation \eqref{quasi_ind} holds with $(Y_1,Y_2)$ in the numerator being replaced by $(Y_1^+,Y_2^+),(Y_1^+,Y_2^-),(Y_1^-,Y_2^+)$, where $Y_{i}^+=\max\left(Y_i,0\right)$ and $Y_{i}^-=\max\left(-Y_i,0\right)$ for $i=$1, 2.\\
In this case, one can show that to prove $Y_1$ and $Y_2$ are asymptotically tailed independent, it suffices to prove that
$P(Y_i^+>t, Y_j^+>t)$, $P(Y_i^+>t, Y_j^->t)$,$P(Y_i^->t,Y_j^+>t)$ are all $o(P(Y_1>t))$ or $o(P(Y_2>t))$. 
\subsection{Proof of Theorem 1}
Before proving Theorem 1, first we introduce two lemmas, Lemma \ref{lemma 1.1} and \ref{lemma 1.2}.
\begin{lemma}\label{lemma 1.1} If $X_1$ and $X_2$ are bivariate standard normally distributed with correlation $-1<\rho<1$, then $|X_1|$ and $|X_2|$ are
asymptotically tailed independent.
\end{lemma}
\begin{proof}Use the upper bound for upper tailed probability of bivariate standard normal random variables. $P(X_1>t, X_2>t)\le \Phi(-t)\Phi(-\theta t)(1+\rho)$ for $t>0$ and $\rho>0$, where $\theta=\sqrt{\frac{1-\rho}{1+\rho}}$ \citep{willink2005bounds}.
We first assume $\rho>0$. When $\rho<0$, let $Z_2=-X_2$. Then $X_1$ and $Z_2$ are bivariate standard normally distributed with correlation $\rho>0$ and 
$P(|X_1|>t,|X_2|>t)=P(|X_1|>t,|Z_2|>t)$. So it suffices to prove the case of $\rho>0$. Now we consider the case where $\rho>0$,
\begin{align}
   &P(|X_1|>t,|X_2|>t) \notag \\
   &\le  P(X_1>t,X_2>t)+P(-X_1>t,-X_2>t)+P(X_1>t,-X_2>t)+P(-X_1>t,X_2<t)\notag \\
   &=I+II+III+IV.\notag
\end{align}
For $I$, we have
$I=P(X_1>t,X_2>t)\le \Phi(-t)\Phi(-\theta t)(1+\rho)=o(P(X_1>t))$.
For $II$, we note
$II=I$ ($X_1$ and $X_2$ are bivariate standard normal random variables, so their joint pdf are symmetric around 0).
For $III$, first let $X_2=c_1X_1+c_2Z$, where $c_1>0$ (because $\rho>0$) and $c_2>0$ and $Z$ is a standard normal random variable independent of $X_1$. Then we have

\begin{align*}
  P(X_1>t,-X_2>t) &= P(X_1>t,-c_1X_1-c_2Z>t) \\
  &= P(X_1>t,-c_2Z>t+c_1X_1) \\
   &\le  P(X_1>t,-c_2Z>t+c_1t)\\
   &= P(X_1>t)P(-c_2Z>t+c_1t)=o(P(|X_1|>t)).
\end{align*}

We then further note $IV=III$ since $X_1$ and $X_2$ are exchangeable. Combine all the results, we have $P(|X_1|>t,|X_2|>t)=o(P(|X_1|>t))$.
\end{proof}

\noindent \textbf{Remark A1.1}: From the Willink's upper bound for bivariate normal r.v.s., it is clear that when $\rho$ is close to 1, we can see the "asymptotically tailed independence phenomenal" only when $t$ is extremely large.

\begin{lemma}[\cite{chen2009sums}]\label{lemma 1.2} If $U_1,\ldots,U_n\in R_{-\gamma}$ are asymptotically tailed independent random variables with CDFs $F_1,\ldots,F_n$, respectively; then $P(U_1+\ldots+U_n>t)\sim\sum_{i=1}^n\bar{F_i}(t)$.\end{lemma}

\begin{proof}[\textbf{\underline{Proof of Theorem 1}}] First we assume the transformation $g(p)$ is nonnegative. Since $U_i\in R_{-\gamma}, \forall i=1,\ldots,n$, by Lemma \ref{lemma 1.2}, it suffices to prove $U_1, \ldots, U_n$ are pairwise asymptotically tailed independent.  Here we have
\begin{align}\label{thm1proofeq1}
P(U_i>t|U_j>t)&=P(g(p_i)>t|g(p_j)>t)\notag\\
&=P(|X_i|>t^*||X_j|>t^*)\sim o((P(|Xi| > t^*))=o(P(U_i>t)).  
\end{align}
Note that $t^*\rightarrow \infty$ as $t\rightarrow \infty$. 
The second equality is because $g(p)$ and $2(1-\Phi(|X|))$ are both monotone decreasing and continuous. $P(|X_i|>t^*||X_j|>t^*)\sim o((P(|X_i| > t^*))=o(P(U_i>t))$ is because of Lemma \ref{lemma 1.1}.
Therefore, $U_1,\ldots,U_n$ are pairwise asymptotically tailed independent and we complete the proof.
When the transformation $g(p)$ is not nonnegative, see Remark A1.2 for detailed proof.
\end{proof}

\noindent \textbf{Remark A1.2:} As described in the proof, we prove Theorem 1 by assuming the transformation $g(p)$ is nonnegative. In fact, it can be easily extended to real-valued transformation $g(p)$. 
In order to prove the asymptotically tailed independence for the general case, it suffices to prove that $P(U_i^+>t, U_j^+>t)$, $P(U_i^+>t, U_j^->t)$, $P(U_i^->t,U_j^+>t)$ are all $o(P(U_i>t))$ or $o(P(U_j>t))$ as $t \rightarrow \infty$.

First for any $t>0$, $P(U_i^+>t,U_j^+>t)=P(U_i>t,U_j>t)$. We can show that $P(U_i>t,U_j>t)=o(P(U_i>t))$ with the same argument as in (\ref{thm1proofeq1}). 
Therefore $P(U_i^+>t,U_j^+>t)=o(P(U_i>t))$. It remains to prove $P(U_i^+>t,U_j^->t)=o(P(U_i>t))$ since $P(U_i^->t,U_j^+>t)=o(P(U_j>t))$ can be proved similarly. 

First we have $P(U_i^+>t,U_j^->t)=P(U_i>t, -U_j>t)=P(U_i>t,U_j<-t)$ for $\forall t>0$. It suffices to show the result hold for the condition (A2.1) in Theorem 1, otherwise for the alternative condition (A2.2), since $U_j$ is bounded below, we have $P(U_j^->t)=P(U_j<-t)=0$ for large enough $t$, which immediately implies $P(U_i^+>t,U_j^->t)=0$. Now we consider the condition (A2.1), where $g(p)$ is continuous and strictly decreasing for $0<p<1$. Note that for any large fixed $t$, there exist a corresponding large fixed value $s_1$ and a small fixed value $s_2$, such that
\begin{align*}
    \left\{U_i>t\right\}&=\left\{|X_i|>s_1\right\}\\
      \left\{U_j<-t\right\}&=\left\{|X_j|<s_2\right\}.
\end{align*}
Because $X_i$ and $X_j$ are bivariate normal distributed with correlation $|\rho_{ij}|\neq 1$, we let
$X_i=C_1Z+C_2X_j$, where $C_1$ and $C_2$ are some constants, $Z\overset{D}{\sim} N(0,1)$ and independent of $X_j$, and then applying similar trick in the proof of Lemma \ref{lemma 1.1}:
\begin{align*}
    P(U_i^+>T, U_j^->t)&=P(|X_i|>s_1, |X_j|<s_2)\\
    &\le P(|C_1Z|+|C_2X_j|>s_1, |X_j|<s_2)\\
    &\le P(|C_1Z|>s_1-|C_2|s_2, |X_j|<s_2)\\
    &= P(|C_1Z|>s_1-|C_2|s_2)P(|X_j|<s_2)=o(P(|X_j|<s_2))=o(P(U_j^->t))
\end{align*}
note $P(U_j^->t)=O(P(U_j>t))$ by balance condition (A3). Hence we complete the proof.\\\\


\subsection{Proof of Theorem 2:}
\begin{proof}[\textbf{\underline{Proof of Theorem 2}}]
First we prove $w_iU_i$ and $w_jU_j$ for $\forall m+1\le i<j\le n$ are asymptotically tailed independent, where the corresponding $|\rho_{ij}|<1$ for $\forall m+1\le i<j\le n$.
As discussed in the Remark A1.2 for Theorem 1, without loss of generality, we can assume both $U_i$ and $U_j$ are nonnegative random variables.
Suppose $w_i\le w_j$:
\begin{align*}
 P(w_iU_i>t|w_jU_j>t) &=\frac{P(w_iU_i>t,w_jU_j>t)}{P(w_jU_j>t)}  \\
   &\le  \frac{P(w_jU_i>t,w_jU_j>t)}{P(w_jU_j>t)}\rightarrow 0  .
\end{align*}

The last line is because $U_i$ and $U_j$ $\forall m+1\le i<j\le n $ are asymptotically tailed independent which were already proved in Theorem 1.\\
Suppose $w_i>w_j$:
\begin{align*}
  P(w_iU_i>t|w_jU_j>t) &=\frac{P(w_iU_i>t,w_jU_j>t)}{P(w_jU_j>t)} \\
   &\le \frac{P(w_iU_i>t,w_iU_j>t)}{P(w_jU_j>t)}   \\
   &= \frac{P(w_iU_i>t,w_iU_j>t)}{P(\frac{w_j}{w_i}w_iU_j>t)}\\
   &\sim  \frac{P(w_iU_i>t,w_iU_j>t)}{(\frac{w_j}{w_i})^{\gamma}P(w_iU_j>t)}\rightarrow 0
\end{align*}
The last line is because $U_i$ and $U_j$ $\forall m+1\le i<j\le n $ are asymptotically tailed independent and also
because $w_iU_j$ has regular-varying tail with index $\gamma$.\\
Hence we have
\begin{equation*}
    \frac{P(w_iU_i>t,w_jU_j>t)}{P(w_iU_i>t)+P(w_jU_j>t)}\le P(w_iU_i>t|w_jU_j>t)\rightarrow 0.
\end{equation*}
Therefore, $w_iU_i$ and $w_jU_j$ $\forall m+1\le i<j\le n$ are asymptotically tailed independent.

Second, we consider the case with extreme correlation $|\rho_{ij}|=1$. In this case, $X_1=...=X_m$ with probability 1 and hence $U_1=...=U_m$ with probability 1. Therefore, it suffice to show that $(\sum_{i=1}^m w_i)U_1$ and $w_jU_j$,  for $\forall m+1\le j\le n$, are asymptotically tailed independent, since 
$\rho_{i j}=1 \text { or }-1 \text { for } 1 \leq i<j \leq m$.

This can be easily proved by the following inequality:
\begin{align*}
   P\left(\left(\sum_{i=1}^m w_i\right)U_1>t|w_jU_j>t\right)\le\sum_{i=1}^m P(w_iU_1>t/m|w_jU_j>t)\rightarrow 0.
\end{align*}


Therefore, 
\begin{align*}
  P(T_{n,w}(\mathbf{X})>t) &= P(\sum_{i=1}^n w_iU_i>t) \\
   &= P\left(\left(\sum_{i=1}^m w_i\right)U_1+\sum_{i=m+1}^n w_iU_i>t\right) \\
  &\sim\left(\sum_{i=1}^m w_i\right)^{\gamma}P(U_1>t)+\sum_{i=m+1}^n w_i^{\gamma}P(U_i>t)\\
  &=\left[\left(\sum_{i=1}^m w_i\right)^{\gamma}+\sum_{i=m+1}^n w_i^{\gamma}\right]P(U_1>t).
\end{align*}
The third line is because $\left(\sum_{i=1}^m w_i\right)U_1$ and $w_jU_j$, $\forall m+1\le j\le n$, are asymptotically tailed independent and Lemma \ref{lemma 1.2} and because of the property of regularly-varying tailed random variables.
\end{proof}

\subsection{Proof of Theorem 3}
Before proving Theorem 3, we first introduce two lemmas for the proof.
Lemma \ref{Lemma 3.1} is the combination of Theorem 2 and Theorem 3 in \cite{davis1983stable}. 
Below are the conditions for Lemma \ref{Lemma 3.1}:\\
\textbf{(B1)}: Let $U_1^*,\ldots,U_{n^*}^*,\ldots$  stationary sequence of regularly-varying random variables with index $0<\gamma\le 1$ and with common distribution function $F^*$.\\
\textbf{(B2)}: Let $G^*(t)=P(|U_1^*|>t)$. The distribution of $U_1^*$ satisfies the balance condition; that is,
$\frac{1-F^*(t)}{G^*(t)}\rightarrow p$ and $\frac{F^*(-t)}{G^*(t)}\rightarrow q$ as $t\rightarrow\infty$,
where $0\le p\le 1$. and $p+q=1$.\\


In addition to conditions $(B1)$ and $(B2)$, there are three additional conditions $(D)$, $(D')$ and $(D'')$ given in \cite{davis1983stable}, all of which
are assumptions for dependent structure of $U_1^*,\ldots,U_{n^*}^*$, and are required for Lemma \ref{Lemma 3.1}. For the details for conditions $(D)$, $(D')$ and $(D'')$, see \cite{davis1983stable}. We do not provide details of these conditions because they are very technical but obviously satisfied in Theorem 3, as shown in the proof of Theorem 3.

\begin{lemma}[\cite{davis1983stable}]\label{Lemma 3.1}
Suppose conditions (B1), (B2), (D), (D') and (D'') hold. For $0<\gamma\le1$ we have
\[\frac{\sum_{i=1}^{n^*}U^*_i-b_{n^*}}{a_{n^*}}\rightarrow_d S^*_{\gamma},\]
where $S^*_{\gamma}$ is a random variable; $a_{n^*}$ is a term such that ${n^*}G^*\left(a_{n^*}x\right)\rightarrow x^{-\gamma}$ for $0<\gamma\le1$ as ${n^*}\rightarrow\infty$ and $x>0$; $b_{n^*}$ is defined as follows 
$$
b_{n^*}=
\begin{cases}
0, & 0<\gamma<1,\\
{n^*}\int_{-a_{n^*}}^{a_{n^*}}xdF^*(x), & \gamma=1,
\end{cases}
$$

\end{lemma}
The following lemma describes the order of $a_{n^*}$ and $b_{n^*}$ given that some of the conditions of Theorem 3 are satisfied.
\begin{lemma}\label{sublemmaofS3}
If $G^*$, $F^*$ and $U_i^*$ for $i=1,\ldots,n$ satisfy conditions for Lemma \ref{Lemma 3.1} and conditions (A3) and (C2), we have
\begin{align*}
    a_{n^*}&=O(({n^*})^{1/\gamma}L_{n^*})\text{ for $0<\gamma\le 1$}\\
    b_{n^*}&=O({n^*}L_{n^*}) \text{ for $\gamma=1$},
\end{align*}
where $L_{n^*}$ is the power function of $\log n^*$.
\end{lemma}
\begin{proof}

First, we prove
$a_{n^*}=O(({n^*})^{1/\gamma}L_{n^*})$ for $0<\gamma\le 1$.
Suppose $a_{n^*}\neq O(({n^*})^{1/\gamma}L_{n^*})$. Then for any $k>0$, there exits an arbitrary large $n^*$, such that $a_{n^*} > ({n^*})^{\frac{1}{\gamma}}\log^k (n^*)$. Hence we have
\begin{align}
    n^* G^*(a_{n^*}x)&\leq n^* G^*\left((n^*)^{\frac{1}{\gamma}}\log^k (n^*)x\right)\notag\\
    &\leq \frac{C n^* \left(\log\left((n^*)^\frac{1}{\gamma} \log^k (n^*)x\right)\right)^h}{\left((n^*)^{\frac{1}{\gamma}}\log^k (n^*)x\right)^\gamma}\notag\\
    &= \frac{C}{x^\gamma}\cdot\frac{\left(\frac{1}{\gamma}\log(n^*)+k\log\log n^* +\log x\right)^h}{(\log n^*)^{k\gamma}},\label{anorder}
\end{align}
where C and h are some fixed constants. The second inequality is due to conditions (A3) and (C2). Indeed, given the two conditions, we have $G^*(t)\stackrel{\text{(i)}}{\le} C \bar{F}^*(t)\stackrel{\text{(ii)}}{\le} \frac{C(\log (t))^h}{t^\gamma}$, where (i) is due to balance condition (A3) and (ii) is due to condition (C2). 
By choosing $k$ such that $k\gamma >h$, we have (\ref{anorder}) $\rightarrow 0$ for $\forall x>0$, which immediately leads to contradiction since by definition of $a_{n^*}$ we have $n^* G^*(a_{n^*}x)\rightarrow \frac{1}{x^\gamma}$.

Then we prove $b_{n^*}=O(n^* L_{n^*})$ for $\gamma=1$. Since conditions (A3) and (C2) hold, we can choose a large enough constant $M$, such that, 
\begin{align*}
&\bar{F}^*(t)\le \frac{(\log(t))^h}{t} \text{ for $\forall t>M$.}\\
    &F^*(-t)\le c\bar{F}^*(t),
\end{align*}
where $c$ and $h$ are fixed some constants. By the definition of $b_{n^*}$, we have
\begin{align*}
    b_{n^*}= n^*\int_{-a_{n^*}}^{a_{n^*}}x d F^*(x)=\underbrace{ n^*\int_{-a_{n^*}}^{-M}x d F^*(x)}_{I}+\underbrace{n^*\int_{-M}^{0}x d F^*(x)}_{II}+\underbrace{n^*\int_{0}^{M}x d F^*(x)}_{III}+\underbrace{n^*\int_{M}^{a_{n^*}}x d F^*(x)}_{IV}
\end{align*}
For $II$ and $III$, we have
$II\le n^*\int_{-M}^0 MdF^*(x)\le n^*M=O(n^*)$ and  $III\le n^*\int_{0}^M MdF^*(x)\le n^*M=O(n^*)$. For $I$, we have 
\begin{align*}
    I= n^*\int_{-a_{n^*}}^{-M} xdF^*(x)=\underbrace{n^*(-M)F(-M)+n^*a_{n^*}F(-a_{n^*})}_{\text{(i)}}-\underbrace{n^*\int_{-a_{n^*}}^{-M} F^*(x)dx}_{\text{(ii)}},
\end{align*}
where (i) is $O\left(n^*L_{n^*}\right)$ since by (A3) we have $n^*a_{n^*}F(-a_{n^*})\le a_{n^*} c n^*\bar{F}^*(a_{n^*})\le c_1 a_{n^*}n^*G^*(a_{n^*})=O\left(n^*L_{n^*}\right)$, where the last equality is due to the fact that $n^*G^*(a_{n^*}x)\rightarrow \frac{1}{x}$ for any $x>0$ and $a_{n^*}=O({n^*}L_{n^*})$ when $\gamma=1$. For (ii), we have 
\begin{align*}
    \text{(ii)}=n^*\int_{-a_{n^*}}^{-M} F^*(x)dx=n^*\int^{a_{n^*}}_{M} F^*(-y)dy&\le n^*\int^{a_{n^*}}_{M} c\bar{F}^*(y)dy\\
    &\le n^*\int^{a_{n^*}}_{M} c \frac{(\log y)^h}{y}dy\\
    &=O\left(n^*(\log(a_{n^*}))^{h+1}\right)=O(n^*L_{n^*}).
\end{align*}
Hence we have $I=O(n^*L_{n^*})$. For $IV$, we have
\begin{align*}
    \left|IV\right|=\left|n^*\int_{M}^{a_{n^*}}xdF^*(x)\right|&=\left|n^*\int_{M}^{a_{n^*}}xd(1-\bar{F}^*(x))\right|=\left|n^*\int_{M}^{a_{n^*}}xd\bar{F}^*(x)\right|\\
    &=\left|n^*a_{n^*}\bar{F}_{a_{n^*}}-n^*M\bar{F}^*(M)-n^*\int_{M}^{a_{n^*}}\bar{F}^*(x)dx\right|\\
    &\le \left|n^*a_{n^*}\bar{F}_{a_{n^*}}\right|+\left|n^*M\bar{F}^*(M)\right|+\left|n^*\int_{M}^{a_{n^*}}\bar{F}^*(x)dx\right|\\
    &\le \left|n^*\int_{M}^{a_{n^*}}\frac{(\log(x))^h}{x}dx\right|+O(n^*L_{n^*}),
    \end{align*}
where the last inequality is due to the fact $n^*a_{n^*}\bar{F}^*(a_{n^*})\le c_1a_{n^*}n^*G^*({a_{n^*}})=O(n^*L_{n^*})$ given (A3) and definition of $a_{n^*}$. Also note that $\left|n^*\int_{M}^{a_{n^*}}\frac{(\log(x))^h}{x}dx\right|=\left|n^*\frac{\log(a_{n^*})^{h+1}}{h+1}-n^*\frac{\log(M)^{h+1}}{h+1}\right|=O(n^*L_{n^*})$. Hence we have $IV=O(n^*L_{n^*})$ and further $b_{n^*}=O(n^*L_{n^*})$

\end{proof}
\noindent \textbf{Remark A1.3:} Lemma \ref{Lemma 3.1} and Lemma \ref{sublemmaofS3} suggest that for the regularly varying variables $U_1^*,\ldots U_{n^*}^*$ with index
$0<\gamma\le 1$, $\sum_{i=1}^{n^*} U_i^*=O({n^*}^{1/\gamma}L_{n^*})$.
For example, for $CA$ test, its corresponding $a_{n^*}=\frac{2{n^*}}{\pi}$ and $b_{n^*}=0$; for $HM$ test, $a_{n^*}={n^*}$ and $b_{n^*}={n^*}\ln({n^*})$;
for $BC_{\eta}$ test ($\eta=1/\gamma$, $0<\gamma<1$), $a_{n^*}=({n^*})^{1/\gamma}$.
The distribution of $S^*_{\gamma}$ is dependent on $\gamma$ and described in details in Theorem 2 and Theorem 3 in \cite{davis1983stable}.
For the purpose of this paper, we will only need to use the order of $\sum_{i=1}^nU_i^*$,
which is $O_p(({n^*})^{1/\gamma}L_{n^*})$ ($0<\gamma\le 1$).



%


 Lemma \ref{Lemma 3.2} and \ref{sublemmaofC2}  are useful when characterizing the lower bound of $g(p)$.
\begin{lemma}[ratio inequality of Mill]\label{Lemma 3.2}
For any $x>0$,
\begin{align*}\frac{x}{\phi(x)}\le1/(1-\Phi(x))\le \frac{x}{\phi(x)}\frac{1+x^2}{x^2},\end{align*}
where $\Phi(x)$ and $\phi(x)$ are CDF and pdf of standard normal distribution, respectively.
\end{lemma}

\begin{lemma}\label{sublemmaofC2}
If conditions (A2), (A3) and (C2) hold, then we have
the following two inequalities for the transformation $g(p)$.\\
There exist $p_1>0,C_1>0,k\ge0$ such that for $0<p<p_1$
\[g(p)\ge \frac{C_1}{p^{1/\gamma}|\ln(p)|^{k}}.\]
and there exist $p_2>0, C_2>0$, $k\ge 0$ such that for $p_2<p<1$
\[g(p)\ge \frac{-C_2|\ln(1-p)|^{k}}{(1-p)^{1/\gamma}}.\]
\end{lemma}
\begin{proof}
To prove the first statement.
Let $t=g(p)$, by condition (A2), $g(p)$ is strictly decreasing for small enough $p$, hence $g^{-1}(t)$ exists for large enough $t$ and is also strictly decreasing. 
Note for any large fixed $t$, we have $F(t)=P(g(p)\le t)=P(p\geq g^{-1}(t))=1-g^{-1}(t)$, hence $\bar{F}(t)=g^{-1}(t)$ for large enough $t$ and further $g(p)=\bar{F}^{-1}(p)$ for small enough $p$, where we have $\bar{F}^{-1}(\bar{F}(t))=t$ for large enough $t$. We now prove the first statement by contradiction, assume for any $k>0$, there exists an arbitrary small $p$ such that $g(p)=\bar{F}^{-1}(p)<\frac{1}{p^{1/\gamma}|\log p|^k}$, which leads to the following contradiction:
\begin{align*}
    t=\bar{F}^{-1}\left(\bar{F}(t)\right)&\le \bar{F}^{-1}\left(\frac{1}{t^\gamma|\log t|^h}\right)\\
    &<\frac{\left(t^{\gamma}|\log t|^h\right)^{\frac{1}{\gamma}}}{\left|\log\left(t^{-\gamma}|\log t|^{-h}\right)\right|^k}\\
    &= \frac{t}{|\log t|^{-\frac{h}{\gamma}}\left(\gamma \log t+h\log\log t\right)^k}<t \text{ by choosing large enough k,}
\end{align*}
 where $h\geq 0$ are some fixed constants. The first inequality is due to condition (C2) and that $\bar{F}^{-1}(p)$ is strictly decreasing for small enough $p$. The second inequality is due to our assumption  $g(p)=\bar{F}^{-1}(p)<\frac{1}{p^{1/\gamma}|\log p|^k}$ for an arbitrary small $p$. Given this contradiction, the proof of the first statement is completed.
 
 We then prove the second statement. First note that when $g(p)$ is bounded below, then the statement is trivial. Since condition (A2) hold for $g(p)$, we only need to prove the statement when $g(p)$ is strictly decreasing for $0<p<1$, because it is trivial for the case $g(p)$ is bounded below and one can note $\frac{-C_2|\ln(1-p)|^{k}}{(1-p)^{1/\gamma}}\rightarrow-\infty$ as $p$ goes to one. 
 
 Now we consider the case where $g(p)$ is strictly decreasing for $0<p<1$. In this case, by similar arguments when we prove the first statement, we denote $t=g(p)$ again and easily note that $g^{-1}(t)$ exists and further $g(p)=\bar{F}^{-1}(p)$ for $0<p<1$, where $\bar{F}^{-1}(\bar{F}(-t))=-t$.  
 
 We now prove the second statement by contradiction.  Given this observation and previously defined notations, by assuming for any $k>0$ there exists an arbitrary small $p$ such that $\bar{F}^{-1}(p)<-C_2 \frac{|\log(1-p)|^k}{(1-p)^{1/\gamma}}$, we derive the following contradiction:
 \begin{align*}
     -t= \bar{F}^{-1}(\bar{F}(-t))&=\bar{F}^{-1}(1-F(-t))\\
     &\le \bar{F}^{-1}\left(1-c_3\frac{|\log t|^h}{t^{\gamma}}\right)\\
     &<-C_2\frac{\left|\log c_3\frac{|\log t|^h}{t^\gamma}\right|^k}{\left(c_3 \frac{|\log t|^h}{t^\gamma}\right)^{1/\gamma}}\\
     &= -t \times C_2 \frac{|\log c_3 -\gamma \log t +h \log \log t|^k}{c_3 (\log t)^{\frac{h}{\gamma}}}< -t \text{ by choosing large enough $k$.}
 \end{align*}
 The first inequality is due to the fact that $\bar{F}^{-1}(p)$ is strictly decreasing and the inequality $F(-t)<c_3\bar{F}(t)\le c_3 \frac{|\log t|^h}{t^\gamma}$ for large enough $t$ and some constants $c_3>0$ and $h\geq 0$, which can be proved given conditions (A3) and (C2) hold. The second inequality is due to our assumption $\bar{F}^{-1}(p)<-C_2 \frac{|\log(1-p)|^k}{(1-p)^{1/\gamma}}$. Given this contradiction, the proof of the second statement is completed.
\end{proof}
\noindent \textbf{Remark A1.4:} One can show that some common transformations $g(p)$ previously discussed satisfy the inequalities above. 
Indeed, the Box-Cox transformation $g(p)=\frac{1}{p^{1/\gamma}}$ satisfies condition (C2). For Cauchy's method, since the corresponding transformation $g(p)=\tan\{(0.5-p)\pi\}$ satisfies
$\lim_{p\rightarrow 0}\frac{g(p)}{1/p}=\frac{1}{\pi}$ and $\lim_{p\rightarrow 1}\frac{g(p)}{\frac{-1}{\pi(1-p)}}=1$, it also satisfies condition (C2). For truncated Cauchy method, since $g(p)=\tan \{(0.5-p) \pi\}$ when $p\leq 1-\delta$, again we have $\lim _{p \rightarrow 0} \frac{g(p)}{1 / p}=\frac{1}{\pi}$, also note when $p>1-\delta$, $g(p)=\tan \{(\delta-0.5) \pi\}$, hence $\lim _{p \rightarrow 1} \frac{g(p)}{\frac{-1}{(1-p) }}=0$, we also have truncated Cauchy satisfied.

\begin{proof}[\textbf{\underline{Proof of Theorem 3}}]

For this theorem, we only consider $0<\gamma\le 1$.
Since $\mathbf{X}$ has banded correlation matrix (condition ($C1$)), we can split $U_1,\ldots,U_n$ into $d_0+1$ groups.
Because we are only looking for the order of asymptotic distribution of $\sum_{i=1}^n U_i$,
we can assume $n$ is a multiple of $d_0+1$ and
let $\frac{n}{d_0+1}-1=n_0$.
Let the divided $d_0+1$ groups as
$\{U_1,U_{(d_0+1)+1},\ldots, U_{(d_0+1)n_0+1}\};$ $\{U_2,U_{(d_0+1)+2},\ldots,U_{(d_0+1)n_0+2}\}$; $\ldots;$ $\{U_{d_0+1},U_{(d_0+1)+d_0+1},U_{(d_0+1)n_0+d_0+1}\}$.
For the $i$th group, the random variables $\{U_i,U_{(d_0+1)+i},\ldots, U_{(d_0+1)n_0+i}\}$ are
identically distributed and independent and hence are stationary.
Also, they are random variables with regularly-varying tails with index $\gamma$ that satisfy
conditions $(A2)$ and $(A3)$. Thus conditions $(B1)$ and $(B2)$ hold.
In addition, since they are independent, it is obvious conditions $(D)$, $(D')$ and $(D'')$ in \cite{davis1983stable} for dependent structure hold.
Let $S_i=\sum_{j=0}^{n_0}{U_{j(d_0+1)+i}},i=1,\ldots,d_0+1$. 
Since $d_0$ is fixed, by applying Lemma \ref{Lemma 3.1} and \ref{sublemmaofS3}, we obtain that $S_i$ is $O_p(n^{1/\gamma}L_n)$.
Therefore, $T(\mathbf{X})=\sum_{i=1}^nU_i=S_1+\ldots+S_{d_0+1}$ is also  $O_p(n^{1/\gamma}L_n)$.
Therefore, now it suffices to prove that under alternative hypothesis $H_a$, $\frac{T(\mathbf{X})}{n^{1/\gamma}L_n}$ converges to $\infty$ with probability 1. Note that, 
\begin{align*}
T(\textbf{X})&=\sum_{i=1}^n g(p_i)=\sum_{i=1}^n g(2(1-\Phi(|X_i|)))\\
                 &=\sum_{i\in S}g(2(1-\Phi(|X_i|)))+\sum_{i\in S^c}g(2(1-\Phi(|X_i|)))\\
                 &= \sum_{i\in S}g(2(1-\Phi(|X_i|)))+O_p(n^{1/\gamma}L_n)\\
                 &\ge  g(2(1-\Phi(\max_{i\in S}|X_i|)))+(n^{\beta}-1)g(2(1-\Phi(\min_{i\in S}|X_i|))) +O_p(n^{1/\gamma}L_n),
%
\end{align*}
where $S=\left\{i: \mu_i\neq 0\right\}$ and $S^c$ is the complementary index set of $S$. The equality in the third line is due to Lemma \ref{Lemma 3.1} and \ref{sublemmaofS3}.  We claim that the if the second term $(n^{\beta}-1)g(2(1-\Phi(\min_{\{i\in S\}}|X_i|)))$ in the last line is negative, its
magnitude is $o_p(n^{1/\gamma})$.

Let $\epsilon_n$ be constant such that $\epsilon_n>0$ and $\epsilon_n\rightarrow 0$ as $n\rightarrow\infty$.
We have
\begin{align*}
  P(\min_{i\in S}|X_i|<\epsilon_n) &\le \sum_{i\in S}P(|X_i|<\epsilon_n)=n^{\beta}P(|X_i|<\epsilon_n)  \\
  & =n^{\beta}\{\Phi(\mu_0+\epsilon_n)-\Phi(\mu_0-\epsilon_n)\}\le2\phi(\mu_0-\epsilon_n)n^{\beta}\epsilon_n\le n^{\beta}\epsilon_n.
\end{align*}
Apply Lemma \ref{sublemmaofC2} we have for small value of $\epsilon_n>0$,
\begin{align}
    g(2(1-\Phi(\epsilon_n)))&\ge \frac{-C^{1}|\log(2\Phi(\epsilon_n)-1)|^{k}}{(2\Phi(\epsilon_n)-1)^{1/\gamma}}\label{theorem3eq1}
\end{align}
Note that $2\Phi(\varepsilon_n)-1=2(\Phi(\varepsilon_n)-\Phi(0))=2(\phi(0)\varepsilon_n+o(\varepsilon_n))=\varepsilon_n(1+o(1))$,
then we have
\begin{align*}
   |\log(2\Phi(\epsilon_n)-1)|^{k}&=|\log(\varepsilon_n(1+o(1)))|^{k}\le 2^k|\log  \varepsilon_n|^k\\
   (2\Phi(\epsilon_n)-1)^\frac{1}{\gamma}&=(\varepsilon_n(1+o(1)))^{\frac{1}{\gamma}}\ge 2^{-\frac{1}{\gamma}}\varepsilon_n^{\frac{1}{\gamma}}.
\end{align*}
Then for the right hand side of (\ref{theorem3eq1}), we have
$$(\ref{theorem3eq1}) \ge -\frac{2^kC^1|\log \varepsilon_n|^k }{2^{-\frac{1}{\gamma}}\varepsilon_n^{\frac{1}{\gamma}}}
    =-C^0\frac{|\ln(\epsilon_n)|^{k}}{\epsilon_n^{1/\gamma}},$$
where $C^{1}>0,C^{0}>0$ are constants.
Now we let $\epsilon_n=n^{\beta_0-1}$, where $\beta<\beta_0<1/2$. Then we have
\[P(\min_{i\in S}|X_i|<\epsilon_n)\le n^{\beta}\epsilon_n=n^{\beta+\beta_0-1}=o(1).\]
We also have
\[n^{\beta}g(2(1-\Phi(\epsilon_n)))\ge -C^{0}n^{\beta-(\beta_0-1)(1/\gamma)}|\ln(n^{\beta_0-1})|^{k}.\]
So we prove that $(n^{\beta}-1)g(2(1-\Phi(\min_{\{i\in S\}}|X_i|)))$ is $o_p(n^{1/\gamma})$.

Then it suffices to prove that $\frac{g(2(1-\Phi(\max_{i\in S}|X_i|)))}{n^{1/\gamma}L_n}$ converges to $\infty$ with probability 1.
Let $S_{+}=\{i \in S, \mu_i>0\}$. Denote $X_i=\mu_0+Z_i$ for $i \in S_+$, where $\mu_0=\sqrt{2\tau\log n}$ and $Z_i\overset{D}{\sim} N(0,1)$. Without loss of generality we assume $|S_{+}|\geq s/2$.
Under the assumption of banded correlation for $X_1,\ldots,X_n$, it follows from Lemma 6 in \cite{cai2014two} that $\max_{i\in S_{+}}Z_i\geq \sqrt{2\log|S_{+}|}+o_p(1)$.
Then we have 
\begin{align*}
    \max_{i\in S}|X_i|\geq\max_{i\in S_{+}}|X_i|\ge \mu_0+\max_{i\in S_{+}}Z_i\geq \mu_0+\sqrt{2\log|S_{+}|}+o_p(1).
\end{align*}
Hence we have
\begin{align*}
  &g(2(1-\Phi(\max|X_i|)))\geq\frac{C_1}{(2(1-\Phi(\max|X_i|)))^{\frac{1}{\gamma}}|\log(2(1-\Phi(\max|X_i|)))|^k}\\
  &\geq\frac{C_1}{(1-\Phi(\max|X_i|))^{\frac{1}{\gamma}-\delta}}+o_p(1)\\
  &\geq C_2\max_{i\in S}|X_i|^{\frac{1}{\gamma}-\delta}\exp\{(\frac{1}{\gamma}-\delta)\max_{i\in S}|X_i|^2/2\}+o_p(1)\\
  &\geq C_2(\sqrt{2\log|S_{+}|}+\mu_0)^{\frac{1}{\gamma}-\delta}\exp\{(\frac{1}{\gamma}-\delta)(\log|S_{+}|+\mu_0^2/2+\mu_0\sqrt{2\log|S_{+}|})\}+o_p(1)  \\
  &\geq \exp\{(\frac{1}{\gamma}-\delta)(\log|S_{+}|+\mu_0^2/2+\mu_0\sqrt{2\log|S_{+}|})\}+o_p(1)\\
  &\geq C_3\exp\{(\frac{1}{\gamma}-\delta)(\beta\log(n)+\tau\log(n)+\sqrt{2\tau\log(n)}\sqrt{2\beta\log(n)-2\log(2)})\}+o_p(1)\\
  &\geq C_3\exp\{(\frac{1}{\gamma}-\delta)(\beta\log(n)+\tau\log(n)+\sqrt{2\tau\log(n)}\sqrt{2\beta\log(n)}-\sqrt{2\log(2)})\}+o_p(1)\\
  &\geq C_3\exp\{(\frac{1}{\gamma}-\delta)(\beta\log(n)+\tau\log(n)+\sqrt{2\tau\log(n)}\sqrt{2\beta\log(n)})\}+o_p(\exp{(\frac{1}{\gamma}-\delta)\sqrt{2\tau\log(n)}})\\
  &\geq C_3\exp\{(\frac{1}{\gamma}-\delta)(\log(n)(\sqrt{\beta}+\sqrt{\tau})^2)\}+o_p(\exp{\sqrt{2\tau\log(n)}}) \\
  &= C_3n^{(\frac{1}{\gamma}-\delta)(\sqrt{\tau}+\sqrt{\beta})^2}+o_p(\exp{\sqrt{2\tau\log(n)}}).
\end{align*}
Note that $\delta$ in the second line is a small positive number.
Note that the inequality in the first line is due to Lemma \ref{sublemmaofC2}.
The inequality in the second line is because $|\log(p)|^k$ is smaller than $p^{-\delta}$ for any positive number $\delta$ when $p$ is small and because $\max|X_i|$ goes to infinity with probability 1.
The inequality in the third line is due to Lemma \ref{Lemma 3.2}
Since $\sqrt{\tau}+\sqrt{\beta}>1$, we can choose $\delta$ so small that $(\frac{1}{\gamma}-\delta)(\sqrt{\tau}+\sqrt{\beta})^2>\frac{1}{\gamma}$.
Therefore, the proof is complete.
\end{proof}
\noindent \textbf{Conclusion:} When $\gamma\le 1$ and $0<\beta<1/4$ (very sparse signal), the decision boundary for test statistic $T(\mathbf{X})$ is optimal.
\section{Results related to truncated Cauchy method}
\subsection{Proof of Proposition 4}
\begin{proof}
Define the following random variables,
\begin{align}
Y_i=X_i1(X_i\geq \nu_\delta) +\nu_\delta1(X_i<\nu_\delta)   \;\; i=1,\ldots, n. \notag
\end{align}

Here $X_i's$ identically and independently follow standard Cauchy distribution, and recall that  $\nu_\delta=\tan\left(\pi(\delta-\frac{1}{2})\right)$ for $0<\delta<1$. Define index set $\mathcal{I}=\left\{k: X_k< \nu_\delta \right\}$ and let $m=|\mathcal{I}|$, the cardinality of $\mathcal{I}$, then under the null,  we can rewrite the upper tail probability of truncated Cauchy method's test statistic in the following form:
\begin{align}
P\left(\frac{1}{n}\sum_{i=1}^n Y_i\geq t\right)=\sum_{j=0}^{n}  P\left(\frac{1}{n}\sum_{i=1}^n Y_i\geq t,\; m=j\right).  \notag
\end{align}
Given the above equivalent form, the tail probability can be divided into the two parts below, which will be bounded in the following proof: 
\begin{align}
 I=P\left(\frac{1}{n}\sum_{i=1}^n Y_i\geq t,\; m=0\right),    \notag\\
 II=\sum_{j=1}^{n}P\left(\frac{1}{n}\sum_{i=1}^n Y_i\geq t,\; m=j\right). \notag
\end{align}
For $I$, we have
\begin{align}
 I= P\left(\frac{1}{n}\sum_{i=1}^n Y_i\geq t,\;X_1,\ldots, X_n \geq \nu_\delta\right) \leq  P\left(\frac{1}{n}\sum_{i=1}^n X_i\geq t\right)=P\left(X_1\geq t\right). \notag
\end{align}
For $II$, note that for the terms  $P\left(\frac{1}{n}\sum_{i=1}^n Y_i\geq t,\; m=j\right)\;\; \text{for}\;\; j=1,\dots,n-1$, we have

\begin{align}
 P\left(\frac{1}{n}\sum_{i=1}^n Y_i\geq t,\; m=j\right)&= \binom{n}{j}P\left(\frac{1}{n-j}\sum_{i=1}^{n-j} X_i\geq \frac{nt-j\nu_\delta}{n-j}\;,m=j\right) \notag\\
 &\leq \binom{n}{j}\left(P\left(X_n<\nu_\delta\right)\right)^j P\left(X_1\geq t\right). \notag
\end{align}

Since $t>0$ and $\nu_\delta<0$, $P\left(\frac{1}{n}\sum_{i=1}^n Y_i\geq t,\; m=n\right)=0$. Hence we have, by the binomial theorem,
\begin{align}
    P\left(\frac{1}{n}\sum_{i=1}^n Y_i\geq t\right)=I+II&\leq P(X_1\geq t)+\sum_{j=1}^n\binom{n}{j}\left(P\left(X_n<\nu_\delta\right)\right)^j P\left(X_1\geq t\right)\notag\\
    &\leq P\left(X_1\geq t\right)\left(1+P(X_n<\nu_\delta)\right)^n.\notag
\end{align}

Notice $\tan\left(\pi(\frac{1}{2}-p)\right)$ follows standard Cauchy distribution, hence  $P(X_n<\nu_\delta)=\delta$, then the result follows.
\end{proof}

\label{lastpage}

\end{document}